\documentclass[aps,
               reprint,
               superscriptaddress]{revtex4-2}

\usepackage{amsmath,amssymb,mathtools}
\usepackage{amsthm}
\usepackage{booktabs}
\usepackage{graphicx}% Include figure files
\usepackage{dcolumn}% Align table columns on decimal point
\usepackage{bm}% bold math
\usepackage[dvipsnames]{xcolor}
\usepackage[hidelinks,
            colorlinks=true,
            allcolors=orange
            ]{hyperref}% add hypertext 
\hypersetup{pdfpagemode=UseNone}% hide sidebar of pdf viewer after compiling
\usepackage{siunitx}
\usepackage{physics2}
\usepackage[italicdiff]{physics}
\usephysicsmodule{braket}

\newcommand{\sSz}{{}^1S_0}

\newcommand{\sPo}{{}^1P_1}

\newcommand{\tPo}{{}^3P_1}

\newcommand{\fermibit}{{}^{171}\mathrm{Yb}}
\newcommand{\fermidit}{{}^{173}\mathrm{Yb}}

\newcommand{\bal}{\begin{equation}\begin{aligned}}
\newcommand{\eal}{\end{aligned}\end{equation}}

\theoremstyle{plain}
\newtheorem{thm}{Theorem}

\newtheorem{pro}[thm]{Proposition}

\allowdisplaybreaks[1]
\begin{document}
\preprint{}

\title{Spin-Cat Qubit with Biased Noise in an Optical Tweezer Array}

\author{Toshi Kusano}
\thanks{kusano@yagura.scphys.kyoto-u.ac.jp}
\affiliation{Department of Physics, Graduate School of Science, Kyoto University, Kyoto 606-8502, Japan}

\author{Kosuke Shibata}
\affiliation{Department of Physics, Graduate School of Science, Kyoto University, Kyoto 606-8502, Japan}

\author{Chih-Han~Yeh}
\affiliation{Department of Physics, Graduate School of Science, Kyoto University, Kyoto 606-8502, Japan}

\author{Keito Saito}
\affiliation{Department of Physics, Graduate School of Science, Kyoto University, Kyoto 606-8502, Japan}

\author{Yuma Nakamura}
\affiliation{Department of Physics, Graduate School of Science, Kyoto University, Kyoto 606-8502, Japan}
\affiliation{Yaqumo, Inc., 2-3-2 Marunouchi, Chiyoda-ku, Tokyo 100-0005, Japan}

\author{Rei Yokoyama}
\affiliation{Department of Physics, Graduate School of Science, Kyoto University, Kyoto 606-8502, Japan}

\author{Takumi Kashimoto}
\affiliation{Department of Physics, Graduate School of Science, Kyoto University, Kyoto 606-8502, Japan}

\author{Tetsushi Takano}
\affiliation{Department of Physics, Graduate School of Science, Kyoto University, Kyoto 606-8502, Japan}
\affiliation{The Hakubi Center for Advanced Research, Kyoto University, Kyoto 606-8502, Japan}

\author{Yosuke Takasu}
\affiliation{Department of Physics, Graduate School of Science, Kyoto University, Kyoto 606-8502, Japan}

\author{Ryuji Takagi}
\affiliation{Department of Basic Science, The University of Tokyo, 3-8-1 Komaba, Meguro-ku, Tokyo 153-8902, Japan}
  
\author{Yoshiro Takahashi}
\affiliation{Department of Physics, Graduate School of Science, Kyoto University, Kyoto 606-8502, Japan}
\date{\today}

\begin{abstract}
Bias-tailored quantum error correcting codes (QECCs) offer a higher error threshold than standard QECCs and have the potential to achieve lower logical errors with less space overhead. The spin-cat qubit, encoded in a large nuclear spin-$F$ system, is a promising candidate for bias-tailored QECCs. Yet its feasibility is hindered by the difficulty of performing fast covariant SU(2) rotation with arbitrary rotation angles for nuclear spins and by a lack of noise characterization for gate operations in neutral atom platforms. Here we demonstrate single-qubit controls of $\fermidit$ spin-cat qubits with nuclear spin $I=5/2$ in an optical tweezer array. We implement a covariant SU(2) rotation and non-linear rotations by optical beams and achieve an averaged single-Clifford gate fidelity of $0.961_{-5}^{+5}$. The measurement of the coherence time and spin relaxation time shows that the idling error becomes increasingly biased toward dephasing errors as the magnitude of the encoded sublevel $|m_F|$ increases. Furthermore, we benchmark the noise bias of rank-preserving gates on spin-cat qubits, demonstrating a finite bias of $18_{-11}^{+132}$, in contrast to the case of the two-level system in $\fermibit$, which shows no bias within the experimental uncertainty. Our work demonstrates the feasibility of spin-cat qubits for realizing bias-tailored QECCs, paving the way for achieving hardware-efficient quantum error correction.
\end{abstract}

\maketitle

%%%%%%%%%%%%%%% introduction %%%%%%%%%%%%%%%%%
\section{Introduction}
Toward large-scale quantum computation, experimental efforts have focused on scaling the number of physical qubits~\cite{scholl2021quantum,ebadi2021quantum,Schymik2022situ,Lars2024super,Tao2024lattice,Gyger2024Cont,Norcia2024Ite,Manetsch20246100,Pichard2024Cryo,Li2025Fast,Chiu2025Continuous,zhu2025Yb2400}. In parallel, resource-efficient quantum error correcting codes (QECCs) have been theoretically developed. Particularly, bias-tailored QECCs~\cite{Alieris2008rep,Tuckett2018sur,Higgott2023XY,Bonilla2021XZZX,Darmawan2021XZZX,Huang20233D,Ruiz2025LDPC} can achieve high error thresholds by a simple modification to standard QECCs with a biased noise model towards dephasing errors. This approach holds the potential to attain logical error rates comparable to those of standard QECCs with fewer qubits, reducing the space overhead of fault-tolerant quantum computation (FTQC). Several candidates for qubits exhibiting a biased noise structure, such as erasure qubits~\cite{Cong2022Leak,Wu2022Erasure,Sahay2023Erasure,Kubica2023Erasure} and bosonic cat qubits~\cite{Cochrane1999Cat,Puri2020CX,Darmawan2021XZZX}, have been proposed and are developed on platforms such as superconducting devices~\cite{Chou2024Dual,Grimm2020Kerr,Reglade2024T1,Qing2024FTQC}, trapped ion systems~\cite{Quinn2024High}, and neutral atom systems~\cite{Scholl2023Erasure,Ma2023high,Chow2024Loss,Bluvstein2025Arch,Zhang2025Lev}.
% 103 words

Another promising candidate for a biased qubit is the spin-cat qubit encoded in large spin-$F$ systems~\cite{Omanakuttan2024spin,Kruckenhauser2025Dark}, which is defined as a superposition of the Zeeman sublevels $m_F = \pm F$ (Fig.~\ref{fig:overview}). Recent studies have shown that this encoding scheme exhibits several advantageous properties for FTQC: (1) hopping errors are correctable unless multiple hoppings change the sign of $m_F$~\cite{Omanakuttan2024spin}; (2) a measurement-free correction for these hopping errors is available~\cite{Omanakuttan2024spin,Kruckenhauser2025Dark,debry2025}; and, (3) bit-flip errors are suppressed and the noise structure is biased towards dephasing~\cite{Omanakuttan2024spin}. A significant advantage of the spin-cat approach is its inherent robustness against idling errors when utilizing a nuclear spin-$F$ system, allowing us to simultaneously achieve a long coherence time~\cite{Yang2025Minute,Yu2025SiliconCat} and strong bit-flip error mitigation.

Despite these favorable characteristics, further investigation is required to determine the experimental feasibility of the spin-cat state for FTQC.
One primary challenge is the realization of fast covariant SU(2) rotations~\cite{Yu2025SiliconCat}, which preserves the shape of the Wigner function to ensure fault-tolerance against hopping errors, in nuclear spin-$F$ systems. While physical systems sensitive to magnetic fields can achieve covariant SU(2) rotations using radio-frequency or magnetic fields~\cite{chalopin2018,debry2025} with Rabi frequency smaller than Zeeman splitting, fast covariant SU(2) rotation is challenging in nuclear spin systems with low magnetic sensitivity. Although optical lasers are expected to enable fast rotation operations, their use can distort the shape of the Wigner function due to the tensor lightshift~\cite{chalopin2018, Yang2025Minute}, making the implementation of covariant SU(2) rotations for arbitrary rotation angles non-trivial.

Furthermore, while the biased noise characteristics of spin-cat states against idling errors have been studied in superconducting~\cite{Roy2025, Champion2025} and silicon~\cite{Yu2025SiliconCat} platforms, their biased noise characteristics against single-qubit gate operations have not been shown. Since gate operation errors are the primary error source in neutral atom systems, gate operation errors, rather than idling errors, characterize the biased properties of the spin-cat qubit. Therefore, it is currently unclear whether the spin-cat qubit is feasible for bias-tailored QECCs in this platform.

\begin{figure*}[ht]
    \centering
    \includegraphics[width=\textwidth]{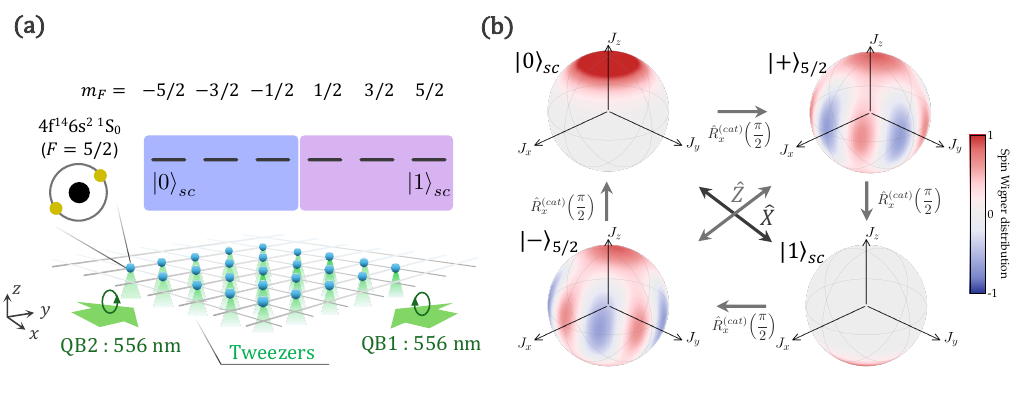}
    \caption{\textbf{Overview of spin-cat state controls in an optical tweezer array.} \textbf{(a)} Schematic illustration of the control beam geometries. The qubit is encoded in the nuclear-spin stretched states of the $\fermidit$ atom ground state as $\ket{0}_{sc}=\ket{-5/2}$ and $\ket{1}_{sc}=\ket{+5/2}$. This encoding scheme suppresses hopping errors by leveraging the redundant states between the encoded qubit states. The spin-cat states $\ket{\pm}_{5/2}$ are generated as superpositions of the stretched states. We control single atoms trapped in an optical tweezer array by using a single-beam Raman technique. QB1 and QB2 are used for spin-cat state preparation and covariant SU(2) rotation, respectively, with different laser detunings. 
    \textbf{(b)} Wigner function representation of the coherent manipulations for the spin-5/2 system. The spin-cat state rotations, denoted as $\hat{R}_x^{(cat)}(\pi/2)$, cyclically map the basis states: $\ket{0}_{sc}\rightarrow \ket{+}_{5/2}\rightarrow \ket{1}_{sc} \rightarrow \ket{-}_{5/2} \rightarrow \ket{0}_{sc}$. The central arrows indicate the application of the $\hat{X}$ gate and the $\hat{Z}$ gate.}
    \label{fig:overview}
\end{figure*}

Additionally, optical tweezer array platforms, offering scalability~\cite{scholl2021quantum,ebadi2021quantum,Schymik2022situ, Lars2024super, Tao2024lattice, Gyger2024Cont, Norcia2024Ite, Manetsch20246100, Pichard2024Cryo, Li2025Fast, Chiu2025Continuous,zhu2025Yb2400}, high-fidelity gates~\cite{Evered2023, Peper2024Spec, Tsai2024Bench, Infleqtion2024Univ, AC2025High, Senoo2025High}, and all-to-all connectivity~\cite{Bluvstein2025Arch}, have primarily been focusing on two-level systems. Utilizing a multi-level system for spin-cat qubit on this platform, together with efficient FTQC schemes with reduced overhead~\cite{zhou2025low, Zhou2025Resource, sunami2025, sahay2025fold,claes2025}, high-threshold QECCs by biased noise~\cite{Bonilla2021XZZX,Darmawan2021XZZX,Sahay2023Erasure} are expected to lower the requirements for utility-scale quantum computation.

Motivated by these considerations, in this paper we demonstrate single-qubit controls of spin-cat qubits encoded in the $\sSz$ ground state of $\fermidit$ atoms with nuclear spin $I=5/2$ trapped in an optical tweezer array. We characterize the gate fidelity and noise bias structure utilizing the Clifford randomized benchmarking (CRB)~\cite{Knill2008,Magesan2011,Magesan2012} and $\mathbb{D}_8$ dihedral randomized benchmarking (DRB)~\cite{Claes2023,Qing2024FTQC}, respectively. The single-qubit gates for the spin-cat qubit is realized by combining three components: optical laser-driven covariant $\text{SU}(2)$ rotation (used as the Pauli gate), non-linear rotation (used as the Hadamard gate), and arbitrary $Z$-axis rotation (including the $\hat{T}$ gate). By employing the covariant SU(2) and non-linear rotations, we achieve high-fidelity single-qubit gates, demonstrated by an averaged Clifford gate fidelity of $0.961_{-5}^{+5}$ using a coarse-grained (CG) measurement~\cite{Kofler2007,Duarte2017} for $m_F < 0$ states. The results indicate that the Clifford gate fidelity improves with increasing CG level, validating the redundancy in the qudit system. Furthermore, the measurement of the dependence of coherence times $T_2^*$ and spin-relaxation times $T_1$ on the encoded Zeeman sublevels reveals that the idling error in the large spin-$F$ system is biased toward the $Z$-error. Finally, we directly characterize the noise bias using the DRB method. Results show a finite bias of $\eta = 18_{-11}^{+132}$ with a non-dephasing error probability of $3.7_{-3.2}^{+3.3}\times10^{-4}$ and a dephasing error probability of $6.7_{-1.5}^{+1.7}\times10^{-3}$, in good contrast to the case of the two-level system of $\fermibit$ which shows no bias within the experimental uncertainty. These measurements establish the feasibility of the spin-cat qubit for realizing bias-tailored QECCs, facilitating the realization of hardware-efficient quantum error correction.

The paper is organized as follows. Section~\ref{sec:SingleBeamRaman} introduces the theoretical framework of single-beam Raman transitions in large-spin systems and presents the necessary and sufficient conditions for implementing high-fidelity single-qubit gates. 
In Sec.~\ref{sec:SingleQubitControl}, we experimentally demonstrate the coherent manipulation of the spin-cat qubit using the single-beam Raman technique and benchmark the average Clifford gate fidelity. 
Section~\ref{sec:Lifetime} characterizes the idling error bias by measuring the coherence and spin-relaxation times for various encoded sublevels. 
In Sec.~\ref{sec:NoiseBias}, we evaluate the noise bias of single-qubit gates using the noise-bias dihedral randomized benchmarking method. 
Section~\ref{sec:ErrorBudget} provides a detailed error budget analysis for the single-qubit gate operations. Finally, Sec.~\ref{sec:Discussion} concludes the paper with a summary and an outlook.

%%%%%%%%%%%%%%%%%%%%%%%%%%%%%%%%%%%%%%%%%%%%%%%%%%%%%%%%%%%%%%
%%%%%%%%%%%%%%%%%%%%%%%%%%%%%%%%%%%%%%%%%%%%%%%%%%%%%%%%%%%%%%
% Single beam raman
%%%%%%%%%%%%%%%%%%%%%%%%%%%%%%%%%%%%%%%%%%%%%%%%%%%%%%%%%%%%%%
%%%%%%%%%%%%%%%%%%%%%%%%%%%%%%%%%%%%%%%%%%%%%%%%%%%%%%%%%%%%%%
\section{Single-beam Raman Transition for a large spin System}
\label{sec:SingleBeamRaman}
In this work, we employ a single-beam Raman technique for multi-spin control, an extension of a technique implemented previously in a two-level system~\cite{Jenkins2022}. The method is reformulated for a six-level system. We assume a control laser that is detuned from the $\sSz \text{-} \tPo$ resonance is applied to the atoms. In this case, the lightshift experienced by $^{173}\text{Yb}$ atoms in the $\sSz$ nuclear-spin manifold can be written as:
\begin{align}
    \label{eq:Lightshift}
    \hat{H}_{LS}/\hbar
    = \sum_{k=0}^{2F}\delta_k\ket{F,m_F=-F+k}\bra{F,m_F=-F+k},
\end{align}
where $\hbar$ is the reduced Planck constant, and $F=5/2$. The magnitude of these lightshifts $\delta_k~(k=0,1,\dots,2F)$ is controllable by the intensity, detuning, and the polarization of the laser. For the details, see Appendix~\ref{method:Lightshift}.

The unitary time evolution for the single-beam Raman transition is written as:
\begin{equation}
    \label{eq:rot_unitary}
    \hat{U}_{rot}^{(F)}(t) = \hat{d}^{(F)}\qty(\beta) \exp\qty(-it \frac{\hat{H}_{LS}}{\hbar}) \hat{d}^{(F)\dagger}\qty(\beta),
\end{equation}
where the $\hat{d}^{(F)}\qty(\beta)$ is the rotational operator. This operator is called the Wigner small $d$-matrix~\cite{JJSakurai2020}:
\begin{equation}
    \label{eq:wigner-small-d}
    d_{m'm}^{(F)}(\beta) = \bra{F,m'}\exp\qty(-i\frac{\beta}{\hbar}\hat{J}_y)\ket{F,m},
\end{equation}
where the rotational angle $\beta$ is the angle between the quantization axis defined by the magnetic field and the control laser, and is set to $90\,^\circ$ ($\beta = \pi/2$) from experimental conditions.
Note that while the original lightshift Hamiltonian Eq.~(\ref{eq:Lightshift}) has only diagonal components, the single-beam Raman operator Eq.~(\ref{eq:rot_unitary}) acquires non-zero off-diagonal components due to the rotational transformation, thereby allowing for coherent transitions between different Zeeman sublevels. Qualitatively, this rotational transformation can be understood as a switching of the quantization axis from the one determined by the magnetic field to the one determined by the lightshift originating from the control laser.
Therefore, the described formulation for the dynamics of this single-beam Raman transition is well justified when the lightshift is sufficiently larger than the Zeeman splitting caused by the magnetic field.
% 192 words

Unlike the two-level system, performing a single-beam Raman transition in a multi-spin system necessitates a \textit{differential lightshift engineering} to control the multi-components differential lightshifts (DLSs) between the Zeeman sublevels, where DLSs are defined by $\Delta_{k+1} = \delta_{k+1}-\delta_k ~(k=0,1,\dots,2F-1)$. To construct single-qubit Clifford gates, we develop $\hat{R}_x(\pi)$ and $R_x^{(cat)}(\pi/2)$ gates given as following,
\begin{align}
    \label{eq:Xgate}
    \hat{R}_x(\pi) &= \sum_{k=0}^{2F}\ket{F,-F+k}\bra{F,F-k}, \\
    \begin{split}
        \label{eq:Hgate}
        \hat{R}_x^{(cat)}(\pi/2) &= \frac{1}{\sqrt{2}}\sum_{k=0}^{2F}\ket{F,-F+k}\bra{F,-F+k} \\
        & \pm\frac{i}{\sqrt{2}}\sum_{k=0}^{2F}\ket{F,-F+k}\bra{F,F-k}.
    \end{split}
\end{align}
The unitary time evolution operator $\hat{U}_{rot}^{(F)}(t)$ for a single-beam Raman transition Eq.~(\ref{eq:rot_unitary}) generally has non-zero matrix elements between any spin states. To make matrix elements zero except for specific components (Eqs.~(\ref{eq:Xgate}) and (\ref{eq:Hgate})) at a certain gate operation time $t$, appropriate DLSs must be chosen. For $F=5/2$, we present the \textit{necessary} and \textit{sufficient} conditions for these DLSs as follows:
\begin{align}
    \label{eq:Xgate_condition}
    &\Delta_1:\Delta_2:\Delta_3:\Delta_4:\Delta_5 \notag \\
    &= (2n_1+1):(2n_2+1):(2n_3+1):(2n_4+1):(2n_5+1),
\end{align}
for the $\hat{R}_x(\pi)$ gate, and
\begin{align}
    \label{eq:Hgate_condition}
    &\Delta_1:\Delta_2:\Delta_3:\Delta_4:\Delta_5 \notag \\
    &= (4n_1\pm1):(4n_2\mp1):(4n_3\pm1):(4n_4\mp1):(4n_5\pm1),
\end{align}
for the $\hat{R}_x^{(cat)}(\pi/2)$ gate. Here $n_k ~(k=1,2,\dots,2F)$ are integers. For the proof, see Appendix~\ref{subsec:OptimalProof}.
%161 words

In addition to analytically deriving the optimal laser detunings, we numerically investigate the optimal laser detunings for the above two types of gates under a more realistic setting that includes the effects of photon scattering, by computing the gate infidelity using a quantum master equation (see Appendix~\ref{subsec:detuningsimulation}). 
We find that there are several detunings where the gate infidelity is significantly suppressed, and confirm that these all satisfy the conditions stated in Eqs.~(\ref{eq:Xgate_condition}) and (\ref{eq:Hgate_condition}).

\begin{figure*}[t]
    \centering
    \includegraphics[width=\textwidth]{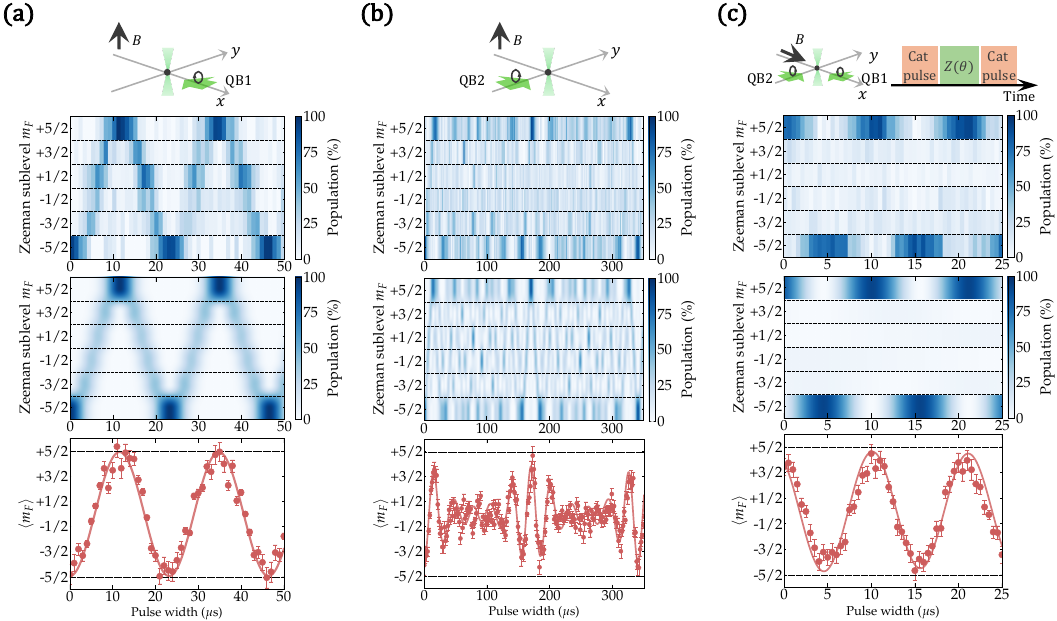}
    \caption{\textbf{Spin-cat qubit manipulations.} Time evolution of the $\ket{m_F}$ state population as a function of pulse duration for \textbf{(a)} a covariant SU(2) rotation, \textbf{(b)} a non-linear rotation, and \textbf{(c)} a $z$-axis rotation. All control lasers are applied in the horizontal plane with circular polarization. QB1 is utilized for the covariant SU(2) and $z$-axis rotations, and QB2 is used for the non-linear rotations. A bias magnetic field is applied orthogonally to the propagation axis of both QB1 and QB2 for the covariant SU(2) and non-linear rotations. In contrast, the magnetic field is aligned parallel to QB1 for the $z$-axis rotation. The simulated dynamics computed from the master equation (middle panels) show good agreement with the experimental data (top panels) for all rotations. The expectation value of the dynamics of the magnetization $\langle m_F \rangle$ (bottom panels) shows (a) a sinusoidal curve with a Rabi frequency of $2\pi\times 43.0\,\text{kHz}$, (b) a beat signal with five distinct frequencies, and (c) a sinusoidal curve with a Ramsey frequency of $2\pi\times 90.5\,\text{kHz}$. For the non-linear rotation, the spin-cat state is generated with a pulse duration of $85.1\,\si{\micro\second}$. In the bottom panels, the solid red lines represent the simulated curves, and error bars represent $1\sigma$ confidence intervals.}
    \label{fig:dynamics}
\end{figure*}

%%%%%%%%%%%%%%%%%%%%%%%%%%%%%%%%%%%%%%%%%%%%%%%%%%%%%%%%%%%%%%
% Single beam raman
%%%%%%%%%%%%%%%%%%%%%%%%%%%%%%%%%%%%%%%%%%%%%%%%%%%%%%%%%%%%%%
\section{Coherent Manipulation \\ of Spin-Cat Qubit}
\label{sec:SingleQubitControl}
The generalized conditions given by Eqs.~(\ref{eq:Xgate_condition}) and (\ref{eq:Hgate_condition}) provide a powerful guideline when we use a multi-spin system for a spin-cat qubit. Here, we experimentally demonstrate the multi-spin dynamics via a single-beam Raman transition and benchmark the average single-qubit Clifford gate fidelity of the spin-cat qubit.
%40words

\subsection{Covariant SU(2) Rotation}
Covariant SU(2) rotations play an essential role in preserving the weight of hopping errors for the spin-cat qubits utilized for FTQC. It fulfills the \textit{rank-preserving} condition against hopping errors, since the occurrence of one hopping error does not induce any additional ones~\cite{Omanakuttan2024spin}.

The control laser for the $\hat{R}_x(\pi)$ gate is circularly polarized and irradiated from a direction perpendicular to the magnetic field (shown as QB1 in Fig.~\ref{fig:dynamics}). This laser is detuned from the $\sSz (F=5/2) - \tPo (F'=7/2)$ resonance by $+11.217\,\text{GHz}$. Under these conditions, the lightshift, which is proportional to $m_F$, acts as a fictitious magnetic field. It thus equalizes the magnitudes of the DLSs between the Zeeman sublevels.

The fact that the lightshift acts as a fictitious magnetic field ensures that this coherent operation is a covariant SU(2) rotation~\cite{Yu2025SiliconCat} (or spin-$F$ SU(2) rotation~\cite{Omanakuttan2024spin}), where the spin state is rotated around a quantization axis while the shape of the Wigner function is preserved. Under this covariant SU(2) rotation, each Zeeman sublevel is most strongly coupled to its nearest neighboring sublevel. We confirmed this behavior experimentally by measuring the population in each sublevel as a function of the pulse duration after initializing the atom in the $\ket{m_F=-5/2}$ state, denoted as $\ket{0}_{sc}$. 
The top panel in Fig.~\ref{fig:dynamics}(a) shows that the population initialized in the $\ket{0}_{sc}$ state is transferred to different Zeeman sublevels as a function of the pulse duration. Moreover, the spin population in the $\ket{0}_{sc}$ is transported to the opposite Zeeman sublevel of $\ket{m_F=+5/2}$, defined as $\ket{1}_{sc}$, at a pulse duration of $11.6\,\si{\micro\second}$, realizing the coherent operation corresponding to the $\hat{R}_x(\pi)$ gate. The experimental results agree well with the simulation (middle panel) using the master equation (see Appendix~\ref{method:Master} for details).

To further characterize the dynamics of the covariant SU(2) rotation, we reconstruct the expectation value of the magnetization $\langle m_F \rangle$ from the spin dynamics data of each Zeeman sublevel (bottom panel in Fig.~\ref{fig:dynamics}(a)). 
The dynamics of $\langle m_F \rangle$ from experiments (dots) show good agreement with the simulation results (solid line).
% The dynamics of the expectation value of the magnetization also show good agreement between the experimental values (dots) and the simulation results (solid line). 
Notably, although the $\fermidit$ atom in the $\sSz$ ground state has six spin components, the magnetization dynamics exhibits a single-frequency sinusoidal curve which is similar to that observed in two-level systems~\cite{Jenkins2022}. This indicates that the realized single-beam Raman transition is a covariant SU(2) rotation. The successful realization of the covariant SU(2) rotations in our experiment demonstrates the feasibility of the spin-cat qubit for FTQC.
% 519 words

%%%%%%%%%%%%%%%%%%%%%%%%%%%%%%%%%%%%%%%%%%%%%%
% Hadamard gate
%%%%%%%%%%%%%%%%%%%%%%%%%%%%%%%%%%%%%%%%%%%%%%
\subsection{Spin-Cat State Generation}
The negative value in the Wigner function of the spin-cat state (Fig.~\ref{fig:overview}(b)) suggests that the cat generation gate $\hat{R}_x^{(cat)}(\pi/2)$ cannot be a covariant SU(2) rotation. To perform this gate, we utilize a different laser beam (shown as QB2 in Fig.~\ref{fig:dynamics}) with circular polarization, and is irradiated from a direction orthogonal to the magnetic field with a detuning of $-5.005\,\text{GHz}$ from the $\sSz (F=5/2) - \tPo (F'=7/2)$ transition frequency. Under this detuning condition, the DLS ratios become $\Delta_1 : \Delta_2 : \Delta_3 : \Delta_4 : \Delta_5 = 7:9:11:13:15$, which satisfies the generalized condition in Eq.~(\ref{eq:Hgate_condition}).

The dynamics of the spin states induced by the spin-cat generation pulse exhibit a non-linear behavior (shown in Fig.~\ref{fig:dynamics}(b)), in the sense that a spin state is coupled not only to the nearest neighboring Zeeman sublevels but also to others. This behavior is in contrast to the covariant SU(2) rotation. From the measurements, atoms initially prepared in the $\ket{0}_{sc}$ state are converted to the spin-cat state,
\begin{equation}
    \ket{+}_{5/2} = \qty(\ket{0}_{sc}+i\ket{1}_{sc})/\sqrt{2},
\end{equation}
at a pulse duration of $85.1\,\si{\micro\second}$, where the pulse corresponds to the $\hat{R}_x^{(cat)}(\pi/2)$ gate. The experimental values of this non-linear rotation and the simulation results computed from the master equation (top and middle panels in Fig.~\ref{fig:dynamics}(b), respectively) also show good agreement.

The complex magnetization dynamics shown in the bottom panel of Fig.~\ref{fig:dynamics}(b) can be decomposed into sinusoidal curves characterized by five different frequencies. The resulting frequencies are in close agreement with the magnitude of the five DLSs $\Delta_k / (2\pi)$ ($k=1, 2, \dots, 2F$). Moreover, we find that the non-linear rotation speed is characterized by the fundamental frequency $f_k = |\Delta_k / (2\pi \times(4n_k \pm 1))|$ ($k=1, 2, \dots, 2F$), where $f_k$ should ideally have the same value for all $k$. Indeed, the average value of the five $f_k$ calculated from our model is $2.939\,\mathrm{kHz}$, which closely matches the inverse of the experimentally observed $2\pi$ pulse time for the non-linear rotation, $2.934(3)\,\mathrm{kHz}$.
Furthermore, the discussion above suggests that a smaller DLS ratio would be favorable for fast non-linear rotation, as the speed is inversely proportional to the magnitude of the ratio. We believe that our analysis here would help to select the optimal control laser detuning when performing non-linear rotation in other atomic species with larger spin-$F$, such as $^{87}\text{Sr}$~\cite{barnes2022assembly, Ahmed2025}.
% 369 words

%%%%%%%%%%%%%%%%%%%%%%%%%%%%%%%%%%%%%%%%%%%%%%
% Z(theta) gate
%%%%%%%%%%%%%%%%%%%%%%%%%%%%%%%%%%%%%%%%%%%%%%
\begin{figure}[b]
    \centering
    \includegraphics[width=\linewidth]{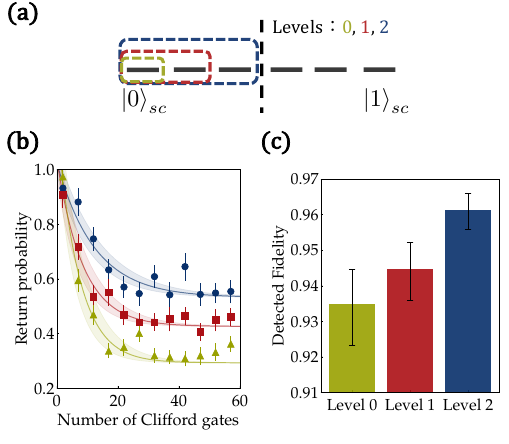}
    \caption{\textbf{Clifford randomized benchmarking (CRB) for the spin-cat qubit.} 
    \textbf{(a)} Schematic illustration of coarse-grained (CG) measurement in the spin-cat encoding. For simplicity, we denote $\ket{^{1}S_{0}, m_F=k}$ as $\ket{k}$, where $k\in \{-5/2, -3/2, -1/2, +1/2, +3/2, +5/2\}$. CG levels are defined by the measured states: (Level~0) only the $\ket{0}_{sc}$ state; (Level~1) the $\ket{0}_{sc}$ and $\ket{-3/2}$ states; (Level~2) the $\ket{0}_{sc}$, $\ket{-3/2}$ and $\ket{-1/2}$ states. 
    \textbf{(b)} Decay of the return probability after CRB circuits, measured using the CG measurements of level~0 (yellow triangle), level~1 (red square), and level~2 (blue circle). Solid curves are fits to the function $ap^m + b_i$, where $m$ is the circuit depth. For each level $i$, we use fixed noise floors $b_i$, which are determined by supplementary experiments on state-selective readout (see Appendix~\ref{method:SSR}). 
    \textbf{(c)} Averaged Clifford gate fidelities extracted from the CRB measurements. The fidelity improves as the CG level increases due to the redundancy in the qudit system.
    We obtain fidelities of $0.935_{-12}^{+10}$, $0.945_{-9}^{+8}$, and $0.961_{-5}^{+5}$ for level~0 (yellow), level~1 (red), and level~2 (blue), respectively. Shaded regions in (b) represent $1\sigma$-confidence intervals of the fit, and error bars in (b) and (c) represent $1\sigma$ confidence intervals.}
    \label{fig:CRBResult}
\end{figure}
\subsection{Phase Control of Spin-Cat State}
A rotation along the $z$-axis is one of the essential gates for universal single-qubit gate operations. In this experiment, we align the direction of the magnetic field to the propagation axis of the control laser beam. This avoids the mixing between different spin states, allowing us to control only the phase of each spin state.

As shown in Fig.~\ref{fig:dynamics}(c), we implement a Ramsey-type experiment and sandwich the $z$-axis rotation pulse with two $\hat{R}_x^{(cat)}(\pi/2)$ pulses to observe the phase dynamics of the spin components. When the initial state is prepared in the $\ket{+}_{5/2}$ state, only the $m_F = \pm 5/2$ states oscillate. The speed of phase oscillation is characterized by the energy difference between the $\ket{0}_{sc}$ and $\ket{1}_{sc}$ state, resulting in a rotation of the spin-cat state phase at an angular frequency of $2\pi \times 90.5\,\text{kHz}$. The corresponding $\hat{R}_z(\pi)$ pulse time is $5.5\,\si{\micro\second}$.
%133 word

%%%%%%%%%%%%%%%%%%%%%%%%%%%%%%%%%%%%%%%%%%%%%%
% Clifford Randomized Benchmarking
%%%%%%%%%%%%%%%%%%%%%%%%%%%%%%%%%%%%%%%%%%%%%%
\subsection{Benchmarking Clifford Gates}
We characterize the fidelity of the single Clifford gates using the basic coherent rotations demonstrated above. The Clifford gates are generated by composing the $\hat{R}_x^{(cat)}(\pi/2)$ and $\hat{R}_z(\pi/2)$ pulses, and the bias magnetic field and laser beam geometry are the same as those depicted in Fig.~\ref{fig:dynamics}(c). To evaluate the averaged Clifford gate fidelity, we employ CRB sequence~\cite{Knill2008,Magesan2011,Magesan2012}. This benchmarking technique measures the dependence of the return probability on the initial state as a function of the circuit depth $m$, where $m$ gates within the Clifford group are randomly sampled.

Unlike a two-level system, the spin-cat encoding involves additional energy levels between the qubit states. These intermediate energy levels provide the spin-cat qubit with redundancy against hopping errors, which are correctable unless they change the sign of $m_F$~\cite{Omanakuttan2024spin,Kruckenhauser2025,debry2025}. To verify this redundancy, we perform CG measurement~\cite{Kofler2007,Duarte2017} at the end of the CRB circuit. For simplicity, we denote from now on $\ket{^{1}S_{0}, m_F=k}$ as $\ket{k}$, where $k\in \{-5/2, -3/2, -1/2, +1/2, +3/2, +5/2\}$. We define CG measurement levels as follows (Fig.~\ref{fig:CRBResult}(a)):
\begin{itemize}
    \item Level~0: Only the $\ket{0}_{sc}$ state is readout.
    \item Level~1: Both the $\ket{0}_{sc}$ and the $\ket{-3/2}$ states are readout.
    \item Level~2: All spin states within the $m_F<0$ manifold are readout.
\end{itemize}
As the CG level increases, the final measurement of the CRB circuit accepts errors that wrongly distribute the spin population to the $\ket{-3/2}$ or $\ket{-1/2}$ state, and thus the detected gate fidelity improves. Our observations demonstrate that the Clifford gate fidelity improves as the CG level increases (Figs.~\ref{fig:CRBResult}(b) and (c)). We obtain the averaged Clifford gate fidelities of $0.935_{-12}^{+10}$, $0.945_{-9}^{+8}$, and $0.961_{-5}^{+5}$ for Level~0 (yellow triangle), Level~1 (red square), and Level~2 (blue circle), respectively.
% 255 words, total 1316 words in this section

%%%%%%%%%%%%%%%%%%%%%%%%%%%%%%%%%%%%%%%%%%%%%%
% Lifetime measurement
%%%%%%%%%%%%%%%%%%%%%%%%%%%%%%%%%%%%%%%%%%%%%%
\section{Lifetime Characterization}
\label{sec:Lifetime}
A unique feature of the spin-cat qubit is its biased noise structure, which originates from the redundant sublevels between the qubit states. In this structure, the bit-flip error is suppressed while the phase-flip error increases as the magnitude of the encoded spin state $\abs{m_F}$ increases, leading to a noise that is biased toward $Z$-error~\cite{Omanakuttan2024spin,Kruckenhauser2025}. To quantitatively characterize the noise structure for idling errors, we measure the coherence time $T_2^*$ and the spin relaxation time $T_1$.

\begin{figure}[t]
    \centering
    \includegraphics[width=\linewidth]{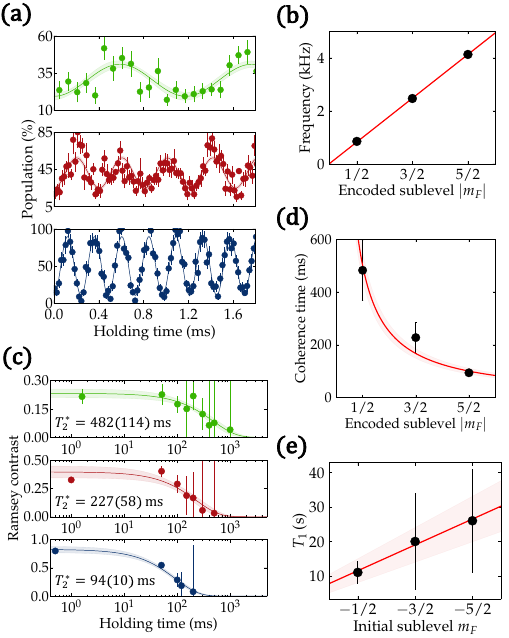}
    \caption{\textbf{Coherence time $T_2^*$ and spin relaxation time $T_1$ measurements.} 
    \textbf{(a)} Ramsey oscillations of the spin-cat and kitten states at a short holding time. The spin-cat state $\ket{+}_{5/2}$ (blue, bottom) and kitten states $\ket{+}_{3/2}$ (red, middle), and $\ket{+}_{1/2}$ (green, top) are prepared, and their respective phase accumulations are measured. 
    \textbf{(b)} Scaling of Ramsey frequency with the encoded sublevel $\abs{m_F}$. The Ramsey frequency extracted from (a) is proportional to the magnitude of the encoded sublevel $\abs{m_F}$, showing a slope of $1.651(6)~\text{kHz}$. 
    \textbf{(c)} Decay of the Ramsey contrasts as a function of holding time $t$, where the $\ket{+}_{5/2}$ (blue, bottom), $\ket{+}_{3/2}$ (red, middle), and $\ket{+}_{1/2}$ (green, top) are initially prepared. The coherence time $T_2^*$ is extracted from a fitting function $\propto \exp\qty(-t/T_2^*)$, yielding $T_2^* = 94(10)\,\text{ms}$ for the spin-cat qubit.
    \textbf{(d)} Scaling of the coherence time $T_2^*$. The coherence time scales with $1/|m_F|$, with a coefficient of $251(21)~\mathrm{ms}$. 
    \textbf{(e)} Scaling of the $T_1$ time. The $\ket{-1/2}$, $\ket{-3/2}$, and $\ket{-5/2}$ states are prepared, and the population of the $m_F>0$ states is measured after a varying holding time at a near-zero magnetic field. We observe that spin relaxation is suppressed for larger $\abs{m_F}$ encoding. Error bars in (a) represent $1\sigma$ confidence intervals, and shaded regions and error bars in (b-e) represent $1\sigma$-confidence intervals of the fit.
    }
    \label{fig:lifetime}
\end{figure}

\subsection{Coherence time measurement}
We evaluate the coherence time using a Ramsey sequence by measuring the dependence of the Ramsey contrast on a variable holding time between the two $\hat{R}_x^{(cat)}(\pi/2)$ pulses.
To characterize the $|m_F|$ dependence on coherence time, we use the spin-cat state $\ket{+}_{5/2}$ and the kitten states $\ket{+}_{3/2}$ and $\ket{+}_{1/2}$ as the initial states, where the kitten states are given as follows:
\begin{align}
    \begin{split}
        \ket{+}_{3/2} = \qty(\ket{-3/2}+i\ket{+3/2})/\sqrt{2}, \\
        \ket{+}_{1/2} = \qty(\ket{-1/2}+i\ket{+1/2})/\sqrt{2}.
    \end{split}
    \label{eq:kitten}
\end{align}
The preparation of these cat and kitten states begins with a state-selective optical pumping using a linearly polarized beam to populate the $\ket{0}_{sc}$, $\ket{-3/2}$, and $\ket{-1/2}$ states, respectively. We note that the subsequently applied $\hat{R}_x^{(cat)}(\pi/2)$ pulse can generate the kitten states with equal pulse duration. The Ramsey oscillations for the spin-cat state and the kitten states at a short holding time are shown in Fig.~\ref{fig:lifetime}(a). The phase of the cat and kitten qubits oscillates at a frequency corresponding to the Zeeman splitting between each qubit state. The Ramsey frequency is proportional to $|m_F|$, with a slope of $1.651(6)\,\text{kHz}$ (Fig.~\ref{fig:lifetime}(b)).

From the time dependence of the Ramsey contrast for a long holding time (Fig.~\ref{fig:lifetime}(c)), we extract the coherence time of the spin-cat state as $T_2^* = 94(10)~\mathrm{ms}$. We compare the coherence time of the spin-cat qubit and that of the kitten qubits and observe that $T_2^*$ is inversely proportional to $\abs{m_F}$, with a coefficient of $251(21)~\mathrm{ms}$ (Fig.~\ref{fig:lifetime}(d)). Although the spin-cat state exhibits enhanced sensitivity to magnetic field fluctuations as shown here, we expect to achieve a coherence time exceeding $10\,\text{s}$, similar to simple two-level systems, by implementing magnetic field stabilization~\cite{Yang2025Minute} or applying dynamical decoupling sequences~\cite{Souza2012}.

\subsection{Spin relaxation time measurement}
In the spin-cat encoding scheme, bit-flip errors, which are uncorrectable for this encoding, originate from hopping errors that change the sign of $m_F$. 
To investigate how the spin relaxation time $T_1$ scales with $|m_F|$, we separately prepared $\ket{0}_{sc}$, $\ket{-3/2}$, and $\ket{-1/2}$ as initial states. Following a variable time delay, we apply a selective pulse to remove atoms in the $m_F > 0$ manifold from the trap.
This selective removal pulse enables us to detect hopping errors into the $m_F>0$ manifold as atomic loss. 

At near-zero magnetic field, the spin relaxation time $T_1$ becomes longer as the magnitude of the $|m_F|$ of the initial state increases (Fig.~\ref{fig:lifetime}(e)). The measured relaxation time $T_1$ follows a linear relation with a slope of $7.5(1.6)\,\text{s}$, reaching $T_1=26_{-10}^{+36}\,\text{s}$ for $m_F=-5/2$. This result substantiates the efficacy of the spin-cat qubit's redundancy against bit-flip errors. Although the demonstrated $T_1$ time for the spin-cat qubit is already sufficiently long, further improvement can be achievable at higher magnetic fields, as previously demonstrated in a two-level system~\cite{Jenkins2022}.

%%%%%%%%%%%%%%%%%%%%%%%%%%%%%%%%%%%%%%%%%%%%%%
% Noise bias 
%%%%%%%%%%%%%%%%%%%%%%%%%%%%%%%%%%%%%%%%%%%%%%
\section{Benchmarking Noise-Bias Structure}
\label{sec:NoiseBias}
In neutral atom quantum processors, gate errors are typically larger than the idling errors. Thus, gate errors, rather than idling errors, should be the primary contribution to the noise bias feature of the spin-cat qubit. To investigate the noise bias characteristic for single-qubit gates, we perform noise-bias DRB sequence~\cite{Claes2023,Qing2024FTQC} to measure the dephasing error probability ($p_D$) and non-dephasing error probability ($p_{ND}$) of the $\mathbb{D}_8$ dihedral gates~\cite{Dugas2015}. These gates belong to single-qubit $\mathbb{D}_8$ dihedral group which is generated by $\hat{X}$ and $\hat{T}$ gates. The $\hat{T}$ gate is a $\pi/4$ rotation around the $z$-axis, defined as $\hat{T} = \ket{0}_{sc}\bra{0}_{sc} + e^{i\pi/4}\ket{1}_{sc}\bra{1}_{sc}$.

\begin{figure}[t]
    \centering
    \includegraphics[width=\linewidth]{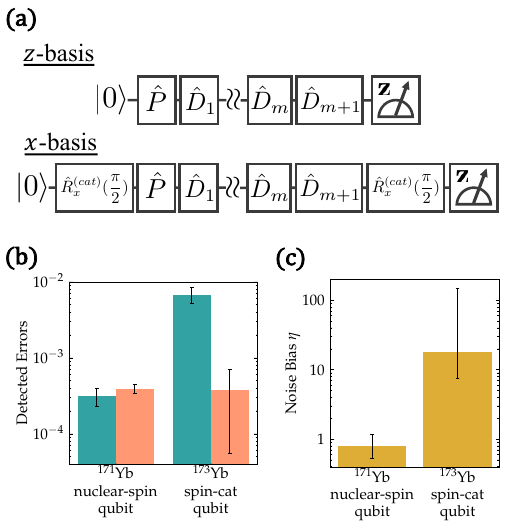}
    \caption{\textbf{Benchmarking noise bias characteristic for single-qubit gates.} 
    \textbf{(a)} Circuit diagrams for the $z$-basis and $x$-basis measurements in the noise-bias dihedral randomized benchmarking (DRB) protocol. 
    \textbf{(b)} Detected non-dephasing errors (orange, right bars in each qubit) and dephasing errors (cyan, left bars in each qubit) extracted from DRB results. For the $\fermidit$ atom, we perform the level~2 CG measurement in the last readout. 
    \textbf{(c)} Noise bias comparison between a nuclear-pin qubit consisting of only two ground sublevels in $^{171}$Yb, and the spin-cat qubit with six Zeeman sublevels in $^{173}$Yb. The quantification of the noise bias, $\eta$, is extracted from the ratio of the dephasing error probability to the non-dephasing error probability. While the nuclear-spin qubit exhibits no noise bias ($\eta=0.8_{-0.3}^{+0.4}$), the spin-cat qubit has a significant and finite noise bias of $\eta=18_{-11}^{+132}$. Error bars represent $1\sigma$ confidence intervals.}
    \label{fig:bias}
\end{figure}

The DRB protocol follows a procedure similar to CRB, with the key difference being the use of random sequences from the $\mathbb{D}_8$ group instead of the Clifford group. For noise-bias characterization, the noise-bias DRB protocol involves two separate experiments: one with preparation and measurement in the $z$-basis and the other in the $x$-basis. The two DRB circuits for the $z$-basis and the $x$-basis measurements are illustrated in Fig.~\ref{fig:bias}(a). By comparing the decay curves from these two circuits, one can separately extract the dephasing and non-dephasing error probabilities, enabling the characterization of biased noise channels. 

The DRB circuit includes a gate $\hat{P}$ randomly sampled from the Pauli group, followed by $m$ gates $\hat{D}_k$ randomly sampled from the $\mathbb{D}_8$ dihedral group, and an inverse gate $\hat{D}_{m+1} = (\hat{D}_m\dots \hat{D}_2\hat{D}_1\hat{P})^{-1}$ to return the population to the initial state. 
To construct all 22 gates in the $\mathbb{D}_8$ dihedral group, we decompose them into the $\hat{R}_x(\pi)$ and $\hat{R}_z(5\pi/4)$ pulse. A complete list of these 22 gates is provided in Appendix~\ref{method:DRB}. 

Figure \ref{fig:bias}(b) shows the extracted non-dephasing and dephasing errors using level~2 CG measurement. 
The extracted non-dephasing and dephasing errors are $p_{ND}=3.7_{-3.2}^{+3.3}\times10^{-4}$ and $p_{D}=6.7_{-1.7}^{+1.5}\times10^{-3}$, respectively, showing that the non-dephasing errors are suppressed more than the dephasing errors by a factor of $ \eta=18_{-11}^{+132}$ (Fig.~\ref{fig:bias}(c)). The measurement results quantitatively demonstrate the existence of a biased noise structure in the spin-cat qubit. Notably, when performing DRB using the $\sSz$ nuclear spin qubit of $\fermibit$ (a two-level system) instead of the spin-cat qubit of $\fermidit$, we find that the non-dephasing and dephasing error probabilities are consistent within the experimental uncertainty, resulting in $\eta = 0.8_{-0.3}^{+0.4}$. This suggests that the noise bias structure is absent for single-qubit gate errors in simple two-level systems, and that the biased noise structure observed in the spin-cat qubit is protected by its redundant sublevels.

\begin{figure}[t]
    \centering
    \includegraphics[width=\linewidth]{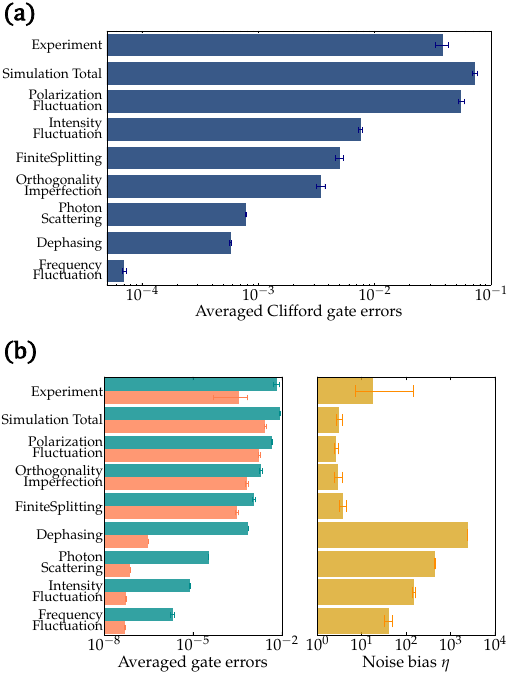}
    \caption{\textbf{Single-qubit gates error budgets.} 
    \textbf{(a)} Clifford gate error budget. The averaged Clifford gate error is analyzed for the level~2 CG measurement. This budget predicts that suppressing the technical imperfections would allow a Clifford gate fidelity exceeding 0.999. These imperfections include polarization and intensity fluctuations of the control laser, the finite Zeeman splitting, and imperfections in the orthogonality between the magnetic field and the laser propagation axis. 
    \textbf{(b)} $\mathbb{D}_8$ dihedral gate error budget and noise bias analysis. The averaged $p_{ND}$ errors (orange, lower bars) and $p_{D}$ errors (cyan, upper bars) are analyzed for the level~2 CG measurement. The experimentally measured noise bias $\eta$ is expected to be limited by dominant technical error sources, whose inherent bias is smaller than the measured result. Near-future experimental upgrades are projected to achieve a noise bias exceeding 100. ``Experiment'' in both figures refers to the averaged gate error derived from either the Clifford randomized or the $\mathbb{D}_8$ dihedral randomized benchmarking. Simulated errors in (a) and (b) are based on experimentally determined parameters (see Appendix~\ref{sec:error_budget_detail}), and the individual errors are added in quadrature to form the total simulated error, denoted as ``Simulation total''. Error bars represent $1\sigma$ confidence intervals.}
    \label{fig:budget}
\end{figure}

\section{Single-qubit gate error analysis}
\label{sec:ErrorBudget}
The spin-cat qubit, possessing the noise bias structure, is expected to reduce the space overhead for FTQC in comparison to standard qubits encoded in unbiased two-level systems~\cite{Alieris2008rep, Tuckett2018sur, Tuckett2019tail, Tucket2020sur, Higgott2023XY, Bonilla2021XZZX, Darmawan2021XZZX, Sahay2023Erasure, Huang20233D, Ruiz2025LDPC}. However, as shown in Fig.~\ref{fig:bias}(b), the spin-cat qubit exhibits a higher single-qubit gate error probability than standard qubits, making the actual benefit provided by the bias structure currently uncertain. 

We anticipate that the spin-cat qubit will hold an advantage over unbiased qubits by considering three aspects and assuming that its single-qubit gate fidelity can exceed $>0.999$. The three aspects are: (1) that the two-qubit gates limit the overall error characteristics of the computation in current neutral atom quantum processors, (2) that the state-of-the-art two-qubit gate fidelity is limited to around 0.999 due to the limited lifetime of the Rydberg state~\cite{Evered2023, Peper2024Spec, Tsai2024Bench, Infleqtion2024Univ, AC2025High, Senoo2025High}, and (3) that the Rydberg state lifetime of the large spin-$F$ system is on the same order as those of simple 2-level qubit systems. In this high-fidelity scenario, a quantum processor utilizing spin-cat qubits would exhibit the same overall physical error probability of $0.999$ as an unbiased two-level system. However, by employing bias-tailored QECCs, it is possible to achieve a higher error threshold than that of the unbiased two-level system, thereby benefiting from the biased noise structure.

To provide a pathway to achieve $>0.999$ single-qubit gate fidelity, we construct an error budget for both the Clifford gates and the $\mathbb{D}_8$ dihedral gates (see Appendix~\ref{sec:error_budget_detail} for details). The current limitation on gate fidelity stems from technical factors, such as shot-to-shot polarization and intensity fluctuations of the control laser, the finite Zeeman splitting effect due to a slow gate control, and the imperfect orthogonality between the applied magnetic field and the laser irradiation axis (Fig.~\ref{fig:budget}(a)). We expect that a fidelity of 0.999 could be realistically achievable through engineering improvements.

The DRB experimental results reveal that the spin-cat qubit exhibits a finite noise bias; however, a larger noise bias is preferable for obtaining a higher error threshold. Figure \ref{fig:budget}(b) shows that the magnitude of the bias is limited by dominant technical noise sources, which have a relatively small bias structure. By suppressing errors originating from these weakly biased noise sources, a larger overall bias is expected to be achieved.

\section{Summary and Outlook}
\label{sec:Discussion}
The demonstration of single-qubit gates for the spin-cat qubit in an optical tweezer array establishes it as a promising candidate for bias-tailored QECCs. Our comprehensive measurements, including lifetime characterization, the Clifford randomized benchmarking, and the $\mathbb{D}_8$ dihedral randomized benchmarking, reveal that the $Z$-biased noise structure is protected by the redundant intermediate energy levels between the qubit states, manifesting a distinct difference from two-level systems.

This work also achieves unique milestones in controlling a large spin-$F$ system. The single-beam Raman technique, generalized for a large spin-$F$ system, provides the condition to engineer high-fidelity single-qubit gates for the spin-cat qubit. Utilizing this generalized formula, we achieve a fast covariant $\text{SU}(2)$ rotation for arbitrary rotation angles using an optical laser beam. This covariant $\text{SU}(2)$ rotation satisfies the rank-preserving condition, highlighting the feasibility of a spin-cat qubit for FTQC.

These results pave the way for hardware-efficient QEC with biased qubits. A remaining gadget required for the practical usage of the biased qubit associated with the rank-preserving $\text{CNOT}$ gate could be implemented via a fast $\hat{X}$ gate while shelving atoms to the Rydberg manifold~\cite{Cong2022Leak,Omanakuttan2024spin}. 
The demonstrated single-beam Raman SU(2) rotation facilitates this feasibility. We also envision that employing erasure conversions~\cite{Sahay2023Erasure} will allow us to use the same gate set as unbiased qubit systems for the spin-cat qubit. This will enable us to achieve both a high-fidelity gate set and a high error threshold. Furthermore, controlling larger spin systems opens the door for exploring novel high-dimensional spin QECCs in single-particle systems~\cite{Victor2020, Gross2021, Jain2024, Aydin2025}.
%total 244 words, whole=

%%%%%%%%%%% Acknowledgment %%%%%%%%%%%
\begin{acknowledgments}
We acknowledge Toshihiko Shimasaki for earlier contributions to the buildout of the optical systems. We thank Takaya Matsuura, Jonathan A. Gross, Shubham P. Jain, Milad Marvian, Vikas Buchemmavari, Ivan H. Deutsch, and Sivaprasad Omanakuttan for insightful discussions. We also thank Koki Ono, Luca Asteria, Amar Vutha, and Sebastian Hofferberth for helpful conversations. This work was supported by Grants-in-Aid for Scientific Research of JSPS (No. JP22K20356, JP24K16975, JP24H00943, JP25K00924), JST CREST (No. JPMJCR1673 and No. JPMJCR23I3), MEXT Quantum Leap Flagship Program (MEXT Q-LEAP) Grant No. JPMXS0118069021, JST Moonshot R\&D (Grants No. JPMJMS2268 and No. JPMJMS2269), JST ASPIRE (No. JPMJAP24C2), JST PRESTO (No. JPMJPR23F5), the Matsuo Foundation, and JST SPRING (Grant No. JPMJSP2110).
\end{acknowledgments}

%%%%%%%%%%% Availability %%%%%%%%%%
\section*{Data availability}
The data that support the plots within this paper and other findings of this study are available from the corresponding author upon reasonable request.

%%%%%%%%%%%%%%%%%%%%%%%% Appendix %%%%%%%%%%%%%%%%%%%%%%%%
\appendix
\section{\MakeUppercase{method}}
\subsection{Trapping and Imaging of $\fermidit$ Atoms}
The experimental sequence is similar to that described in our previous works~\cite{Nakamura2024, Kusano2025}, but here we use the $\fermidit$ isotope. 
After a 0.8\,s loading period into a 3D magneto-optical trap (MOT) performed on the $\sSz $-$ \tPo$ electric dipole transition near 556\,nm with a natural linewidth of about $2\pi\times182\,\text{kHz}$, $\fermidit$ atoms are loaded into an optical tweezer array produced by a laser near 532\,nm (Verdi V-10, Coherent). A phase-only spatial light modulator (X15213-L16, Hamamatsu) and a 0.6-numerical aperture (NA) objective lens (Special Optics) are used to form a $6\times6$ tweezer array with site separation of $6\,\si{\micro\meter}$. To prepare a single atom per tweezer site, we apply light-assisted collision (LAC) beams, the same laser beams as those for the 3D MOT, for 100\,ms with a magnetic field (perpendicular to the tweezer beam propagating axis) of 0.13\,mT and a trap depth of 1.09\,mK. The LAC beams is red-detuned from the $\sSz ~ (F=5/2) $-$ \tPo ~(F'=7/2)$ resonance frequency. After the LAC, we image single $\fermidit$ atoms utilizing the $\sSz $-$ \sPo$ electric dipole transition near 399\,nm with a natural linewidth of about $2\pi\times29\,\text{MHz}$. The 399\,nm and the 556\,nm laser beams irradiate the atoms for 12\,ms, simultaneously. The 556\,nm beams share the same path as the 3D MOT beams and are used to cool the atoms during imaging. The emitted photons are collected by another 0.6-NA objective lens (Special Optics) and subsequently focused onto an electron-multiplying charge-coupled-device camera (iXon-Ultra-897, Andor). After the imaging, the atoms are cooled with the laser beams near 556 nm to $34(2)\,\si{\micro K}$ in a trap depth of  1.09\,mK at nearly zero magnetic field.

\begin{figure}[t]
    \centering
    \includegraphics[keepaspectratio,width=0.9\linewidth]{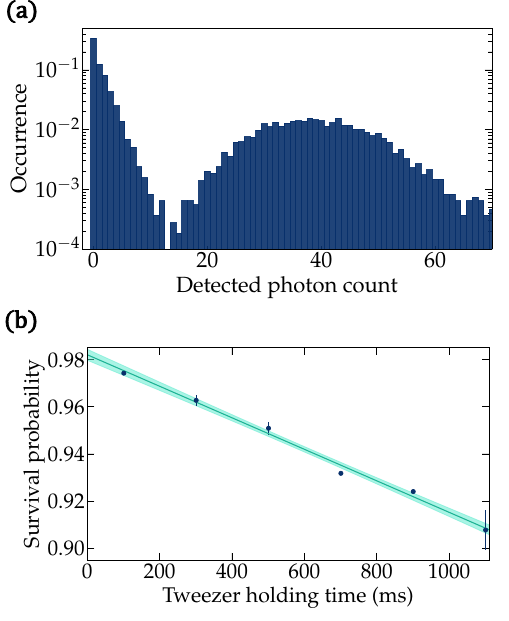}
    \caption[Imaging characterization for $\fermidit$ atom]{
        \label{fig:imaging}
        \textbf{Imaging characterization for $\fermidit$ atom.}\quad 
        \textbf{(a)} Histogram of detected photons from $\fermidit$ with an exposure time of 12\,ms. The discrimination fidelity is 0.99984(6). 
        \textbf{(b)} Extrapolation of the actual survival probability of imaging. Tweezer holding time between two images degrades the detected survival probability. The survival probability at zero holding time is extrapolated to be 0.982(2). The shaded region represents $1\sigma$-confidence intervals from the linear fit, and error bars represent $1\sigma$ confidence intervals from the measurements.
        }
\end{figure}

We characterize the imaging performance using a model-free method~\cite{Manetsch20246100}, which enables us to evaluate the fidelity precisely without imposing any assumptions. With this approach, we obtain a discrimination fidelity of 0.99984(6) with a survival probability of 0.9751(4) under an exposure time of 12\,ms in a 1.09\,mK trap. We find that the holding time during each imaging session causes heating and underestimates the survival probability during the imaging. To obtain the actual imaging survival probability, we measure the dependence of the survival probability on the holding time during each imaging and extrapolate the actual probability at zero holding time with the value of 0.982(2) (shown in Fig.~\ref{fig:imaging}). The measured discrimination fidelity and survival probability for imaging are almost comparable with the values stated in Ref.~\cite{Karim2025single}, where a deeper trap was utilized.

\subsection{Optical Pumping for State-Selective Readout and Initialization}
\label{method:SSR}
The imaging method described above can only determine whether the atoms are in the $^{1}S_{0}$ ground state or not, but is unable to acquire the information on the population of each substate. To perform state-selective readout (SSR) and the initialization of the 6-spin components in the $\sSz$ state, we utilize optical pumping (OP) in a large magnetic field setting (Fig.~\ref{EDfig:SSR}). A magnetic field of 4.6\,mT is applied, resulting in a Zeeman shift of $27.6\times m_{F'}\,\text{MHz}$ within the $\tPo, F'=7/2$ manifold. This shift enables selective optical pumping to the desired Zeeman sublevels in the $\sSz$ manifold via the $\sSz $-$ \tPo ~(F'=7/2)$ transition.

We apply the pump light at a wavelength of 556\,nm with linear polarization oriented perpendicular to the direction of the applied magnetic field.
By applying OP pulses with different frequencies tuned to each $|m_{F'}- m_F| =1$ transition sequentially, we can pump the population into the desired state. The schematic illustration of OP pulse sequences is shown in Fig.~\ref{EDfig:SSR}(a).

By exploiting this selective OP approach, we perform a destructive SSR, where the spin populations other than the target spin state are converted to atomic loss from the trap, and the target spin state remaining in the trap is read out via atomic fluorescence imaging. The pushout beams for introducing the atomic loss mechanism are performed with frequencies on resonance with either the $\sSz ~ (F=5/2, m_F=-5/2)$-$ \tPo ~(F'=7/2, m_{F'}=-7/2)$ transition or the $\sSz ~ (F=5/2, m_F=+5/2)$-$ \tPo ~(F'=7/2, m_{F'}=+7/2)$ transition. To remove population in the spin states other than $\ket{\sSz, F=5/2, m_{F}=\pm5/2}$ states, a selective OP sequence that pumps the population to the $m_F=+5/2$ or $m_F=-5/2$ state is performed before irradiating the pushout beam. The schematic illustration of the SSR sequences is shown in Fig.~\ref{EDfig:SSR}(b).

To address the $|m_{F'}-m_F|=1$ transitions in $\sSz $-$ \tPo ~(F'=7/2)$, which are distributed across a wide frequency range of up to $193.2\,\text{MHz}$ at a magnetic field of 4.6\,mT, we use a double-pass Acousto-Optic Modulator configuration to derive the laser frequency near each corresponding resonance. When performing the selective OPs, each OP pulse is red-detuned by a different value from its corresponding transition frequency to suppress atom loss due to the pumping process. Additionally, the OP pulse for each Zeeman transition is applied with a different pulse width and a different number of iterations. 
The trap is set to a deep depth of $1.09\,\text{mK}$ to suppress atom loss during OP, and is reduced to $0.1\,\text{mK}$ to facilitate atom loss when performing the pushout procedure.
The pushout beam is a single pulse with a duration of $20\,\text{ms}$.

\begin{figure*}[t]
    \centering
    \includegraphics[keepaspectratio,width=\textwidth]{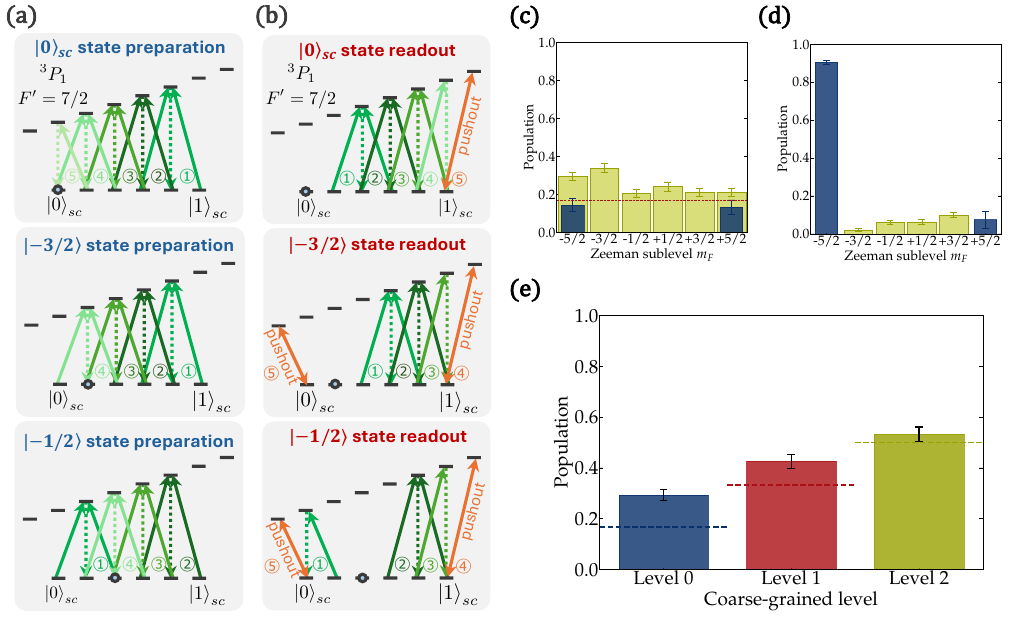}
    \caption[State-selective readout]{
        \label{EDfig:SSR}
        \textbf{State-selective optical pumping}\quad 
        \textbf{(a)} Schematic illustration of state-selective initialization of the $\ket{0}_{sc}$, $\ket{m_F=-3/2}$, and $\ket{m_F=-1/2}$ states. Multiple optical pumping pulses with different frequencies tuned to each $|m_{F'}-m_F|=1$ transition can pump the population to the desired state. 
        \textbf{(b)} Schematic illustration of state-selective destructive readout of the $\ket{0}_{sc}$, $\ket{m_F=-3/2}$, and $\ket{m_F=-1/2}$ states. Spin populations other than the target spin state are converted to atomic loss from the trap. The population in the target spin state is then detected via atomic fluorescence imaging. The state-selective pump beam is linearly polarized, and only resonant transitions by such a beam are shown. 
        \textbf{(c)} Population of all spin states in the $\sSz$ manifold without optical pumping. \textit{Fluorescence (Loss) detection} results are shown as yellow bars (blue narrow bars). The red dashed line represents the population of the spin states within the $\sSz$ manifold when maximally mixed. 
        \textbf{(d)} Population of all spin states in the $\sSz$ manifold after the initialization to the $\ket{0}_{sc}$ state. The measured $\ket{0}_{sc}$ population is 0.904(10). Fluorescence (loss) detection results are shown as yellow (blue) bars. 
        \textbf{(e)} Coarse-grained measurement result without state-initialization. The dashed lines are the expected populations for the maximally mixed $\sSz$ state. These measurements employ fluorescence detection with the error bars representing $1\sigma$ confidence intervals.
        }
\end{figure*}

The results of the state-selective readout without spin state initialization are shown in Fig.~\ref{EDfig:SSR}(c), which is reconstructed from six individual spin state measurements. When performing the destructive SSR, a population larger than the expected value of $1/6$ is observed. This suggests that the deficient selective OP causes states other than the desired readout state to remain in the trap, instead of being pumped to $\ket{\sSz, F=5/2, m_{F}=\pm5/2}$.
Therefore, we perform the \textit{loss detection}, which is capable of state readout without selective OP, by directly irradiating a pushout beam and detecting the atomic state distribution as atom loss. Under this loss detection scheme, the populations of $m_F=+5/2$ and $m_F=-5/2$ are observed to be 0.132(37) and 0.144(34), respectively, agreeing with the ideal uniform distribution of $1/6$ (red-dashed line) within the uncertainty.

To avoid ambiguity, we categorize the state-selective measurement techniques used in this work into \textit{fluorescence detection} and \textit{loss detection}. All measurements for the population in the $m_F=\pm 3/2$ and $\pm 1/2$ states reported in this paper rely on \textit{fluorescence detection} that exploits the destructive SSR. Similarly, all coarse-grained (CG) measurements (Level~0, 1, and 2) shown in Figs.~\ref{fig:CRBResult}, \ref{fig:bias} and \ref{EDfig:SSR}(e) use \textit{fluorescence detection}. The $m_F=\pm 5/2$ state measurements in Figs.~\ref{fig:dynamics}(a-c), \ref{fig:lifetime}(b), and \ref{EDfig:SSR}(d, e) utilize \textit{loss detection}.

Moreover, to accurately determine spin populations, we calibrate the spin-selective readout results by accounting for the survival probabilities during both state-selective initialization and state-insensitive imaging. This calibration is performed by comparing two distinct measurements: one including the state-selective readout pulse and another omitting it. Since both sequences are subject to background atomic loss during the initialization and imaging stages, comparing the survival probabilities from these two cases allows us to decouple the spin population from the atomic loss. This normalization method is applied to all population measurements presented in this work. 

Next, we measure the spin state distribution after initializing the system to state $\ket{\sSz, F=5/2, m_{F}=-5/2}$. In this experiment, we use selective OP to prepare the initial state $m_{F}=-5/2$, followed by state-selective measurement, and the resulting state distribution is shown in Fig.~\ref{EDfig:SSR}(d). The values are also reconstructed from six individual measurements of the spin state. 
To accurately determine the population in the states, fluorescence detection and loss detection are used for measuring the population in the $m_F=\pm 3/2, \pm 1/2$ states and the $m_F=\pm 5/2$ states, respectively.

We observe that 90.4(1.0)\,\% of the population is initialized into the $m_F=-5/2$ state after pumping. For the intermediate states ($|m_F| \leq 3/2$), we expect their true population to be even lower than observed, also due to overestimation caused by the non-optimal selective OP. Furthermore, the $m_F=+5/2$ state retains a finite population of $7.5(4.5)\,\%$ due to pumping error to the $m_F=+3/2$ state.

We emphasize that while spin-selective measurement of multi-spin ensembles has been conventionally performed in cold atom experiments using the Stern-Gerlach technique~\cite{Sleator1992, Taie2010}, this work marks the first attempt to measure multi-spin states on a single atom state-selectively. Although the achieved fidelity for both SSR and initialization is not optimal, it is sufficient for characterizing the multi-spin dynamics, as shown in Fig.~\ref{fig:dynamics}. We anticipate that the demonstrated state-selective method will open the door for exploring $\text{SU}(N)$ physics utilizing multi-spin qudit systems~\cite{taie2022observation,sonderhouse2020} within the optical tweezer array platform.

%%%%%%%%%%%%%%%%%%%%%%%%%%%%%%%%%%%%%%%%%%%%%%%%%%%
% Limitation of state-selective pumping fidelity
%%%%%%%%%%%%%%%%%%%%%%%%%%%%%%%%%%%%%%%%%%%%%%%%%%%
\begin{figure*}[t]
    \centering
    \includegraphics[keepaspectratio,width=\textwidth]{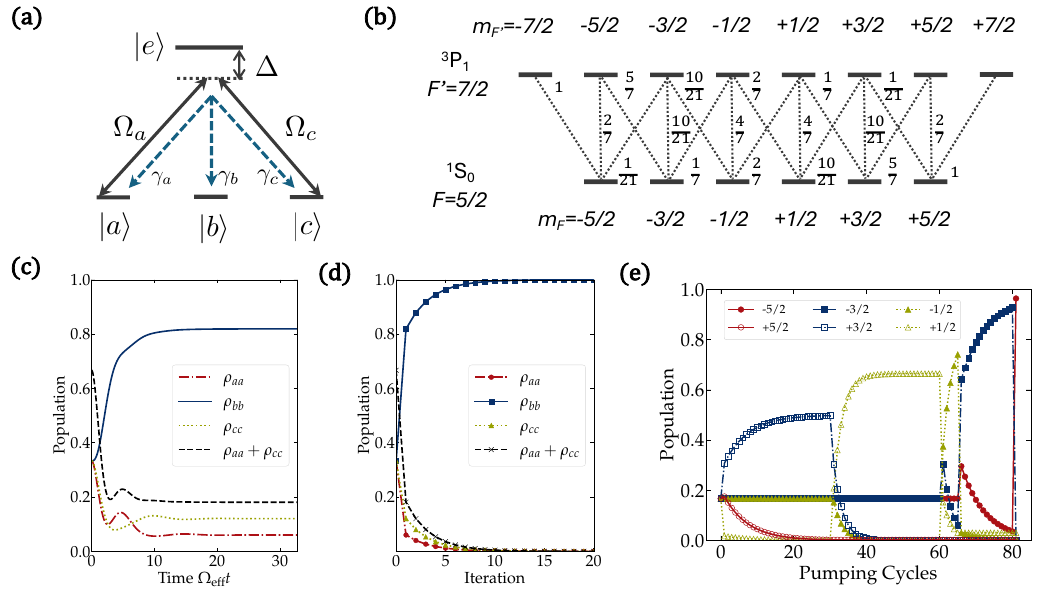}
    \caption{
        \textbf{Selective optical pumping limitation due to dark state generation.}\quad 
        \textbf{(a)} Schematic of the four-level system involved in the selective optical pumping process. The linearly polarized pumping light couples the three states: $\ket{a}$, $\ket{c}$ and $\ket{e}$. A dark state $\ket{D} = (\Omega_c\ket{a}-\Omega_a\ket{c})/\sqrt{\Omega_a^2 + \Omega_c^2}$ does not couple to the pumping light. \textbf{(b)} Squared values of the Clebsch-Gordan coefficients for the $\sSz\,(F = 5/2)$ - $\tPo\,(F' = 7/2)$ transitions. 
        \textbf{(c)} Numerical simulation of the selective optical pumping process for the four-level system spanned by $\ket{a} = \ket{\sSz,m_F=-1/2}$, $\ket{b} = \ket{\sSz,m_F=+1/2}$, $\ket{c} = \ket{\sSz,m_F=+3/2}$ and $\ket{e}=\ket{\tPo,F'=7/2,m_{F'}=+1/2}$ states. In this simulation, we simply use parameters characterized by Clebsch-Gordan coefficients: $\{\Delta,\Omega_a, \Omega_c, \gamma_a, \gamma_b, \gamma_c\} = \{0, \sqrt{2/7}, \sqrt{1/7}, 2/7, 4/7, 1/7\}$. The effective Rabi frequency is given by $\Omega_{\textrm{eff}} = \sqrt{\Omega_a^2+\Omega_c^2}$. 
        \textbf{(d)} Populations after multiple iterations of the selective optical pumping process. When coherence between the ground states is lost, repeating the pumping process multiple times increases the final population of the target state $\ket{b}$. 
        \textbf{(e)} Numerical simulation results of the selective optical pumping process for initializing to the $\ket{0}_{sc}$ state after multiple pumping cycles. The horizontal axis represents the accumulated pumping cycles for each $m_{F'}$ excitation. The OP sequence used here is the same as that in Fig.~\ref{EDfig:SSR}\,(a). In accordance with the experimental parameters associated with Fig.~\ref{EDfig:SSR}(d), the repetitions of OP pulses for each excitation ($m_{F'} = +3/2, +1/2, -1/2, -3/2, -5/2$) are 30, 5, 5, 15, and 1, respectively. The simulated final population in the $\ket{0}_{sc}$ state is 0.964.
        }
        \label{fig:OP_simulation}
\end{figure*}
\subsection{Limitation of state-selective pumping fidelity}
As shown in Fig.~\ref{EDfig:SSR}, discrepancies between the experimental and ideal distributions are observed. We attribute this infidelity of the selective optical pumping to the generation of dark states by the linearly polarized pumping light.

To see this, we consider a simple model where the OP laser frequency is near-resonant with a specific sublevel $m_{F'}$ within the $|m_{F'}|<5/2$ manifold. While the Zeeman shift spectrally resolves the sublevels of the $\tPo$ state at a magnetic field of 4.6\,mT, owing to the low magnetic sensitivity of the ground-state nuclear spin (Zeeman shift: $9.5\,\text{kHz/mT}\times m_F$), the application of linearly polarized pumping light generates a coherent coupling among the three states: $m_F = m_{F'} \pm 1$ and $m_F'$ (Fig.~\ref{fig:OP_simulation}(a)). 

Under these conditions, a dark state that does not couple to the pumping light is generated:
\begin{equation}
    \ket{D} = \frac{1}{\sqrt{\Omega_a^2 + \Omega_c^2}}(\Omega_c \ket{a} - \Omega_a \ket{c}),
    \label{eq:dark_state}
\end{equation}
where $\ket{a} = \ket{\sSz,m_F=m_{F'}-1}$, $\ket{c} = \ket{\sSz,m_F=m_{F'}+1}$, and $\Omega_a$ ($\Omega_c$) is the Rabi frequency between states $\ket{a}$ ($\ket{c}$) and $\ket{e} = \ket{\tPo,m_{F'}}$. This dark state does not include the excited state $\ket{e}$, and thus the population in this state is not pumped to the target state $\ket{b} = \ket{\sSz,m_F=m_{F'}}$. 

The quantitative effect of this dark state is evaluated using the master equation under the following Hamiltonian and collapse operators:
\begin{equation}
    \hat{H} = \hbar\Delta\ket{e}\bra{e} + \hbar\Omega_a(\ket{e}\bra{a} + \ket{a}\bra{e}) + \hbar\Omega_c(\ket{e}\bra{c} + \ket{c}\bra{e})),
\end{equation}
\begin{equation}
    \hat{C}_n = \sqrt{\gamma_n}\ket{n}\bra{e}, \text{ for } n=a,b,c,
\end{equation}
where $\gamma_n$ is the spontaneous decay rate from the excited state to the ground states. 

Figure~\ref{fig:OP_simulation}(c) shows the population dynamics for a four-level system spanned by $\ket{a} = |\sSz, m_F = -1/2\rangle$, $\ket{b} = \ket{\sSz,m_F=+1/2}$, $\ket{c} = \ket{\sSz,m_F=+3/2}$ and $\ket{e}=\ket{\tPo,F'=7/2,m_{F'}=+1/2}$ states. The simulation begins with a uniform distribution across states $\ket{a},\ket{b}$ and $\ket{c}$. As the pumping process proceeds, population is transferred to the target state $\ket{b}$. However, a fraction of the spin population remains trapped in the dark state formed by $\ket{a}$ and $\ket{c}$, resulting in incomplete pumping to the target state $\ket{b}$. In the steady state, the population ratio between the $\ket{a}$ and $\ket{c}$ states is determined by:
\begin{equation}
    \frac{P_a}{P_c} = \frac{\Omega_c^2}{\Omega_a^2} 
                    = \abs{\frac{\langle{F' \, m_{F'} | F \, m_F=m_{F'}+1; 1\, -1 \rangle}}{\langle{F'\, m_{F'} | F\, m_F=m_{F'}-1; 1\, + 1 \rangle}}}^2,
\end{equation}
where the second equality is derived from the Clebsch-Gordan coefficients for the respective transitions (Fig.~\ref{fig:OP_simulation}(b)). Note that the second equality assumes a simple model where state mixing in the excited levels is neglected. For instance, when the tensor lightshift is non-negligible, the actual Rabi frequencies deviate from those characterized by the Clebsch-Gordan coefficients.

While the steady-state ratio $P_a/P_c$ is determined by the dark state,
the overall pumping ratio $P_b / (P_a + P_c)$ depends on the initial distribution when coherence between the ground states is lost. Therefore, by repeating the pumping cycle and allowing the coherence to decay between iterations, the final population pumped to the target state $\ket{b}$ can be increased (Fig.~\ref{fig:OP_simulation}(d)). We experimentally confirmed this improvement in pumping fidelity through repeated selective OP sequences. A single iteration of the sequence shown in Fig.~\ref{EDfig:SSR}(a) yields a population of approximately 80\% in the $\ket{0}_{sc}$ state, whereas multiple iterations achieve a fidelity of 90.4(1.0)\% (Fig.~\ref{EDfig:SSR}(d)).

Finally, we evaluate the contribution of the dark state to the initialization to the $\ket{0}_{sc}$ state. We perform numerical simulations of the selective optical pumping process for the six-level system. 
In accordance with the experimental parameters, we apply the pumping light to each $m_{F'}$ excitation for 1-30 iterations.
The results are shown in Fig.~\ref{fig:OP_simulation}(e), where we plot the population in the $\ket{0}_{sc}$ state after applying the selective OP sequence shown in Fig.~\ref{EDfig:SSR}\,(a) as a function of the number of OP iterations. After the initialization, the population in the $\ket{0}_{sc}$ state reaches 0.964, suggesting that the dark state generation is one of the dominant factors limiting the fidelity of state initialization in our experiment. In near-future experiments, this limitation can be overcome by optimizing the selective OP iterations, or by employing circularly polarized pumping light under a strong magnetic field, which eliminates the formation of dark states.

%%%%%%%%%%%%%%%%%%%%%%%%%%%%%%%%%%%%%%%%%%%%%%
% LightShift
%%%%%%%%%%%%%%%%%%%%%%%%%%%%%%%%%%%%%%%%%%%%%%%%%%%
\section{Single-Beam Raman Transition for a Large Spin-$F$ System}
\subsection{Lightshift by a Detuned Laser}
\label{method:Lightshift}
Here, we provide a detailed formulation of the single-beam Raman transition and introduce the quantum master equation used for numerical simulations in this paper.

First, consider the control laser to be detuned from the $\sSz$-$\tPo$ resonance with an intensity of $I_L$ and polarized in the $\mathbf{e}_L$ direction. The polarization vector is given in the spherical coordinate representation as
\begin{align}
    \begin{split}
        \vb*{e}_L 
        = &\sin2\varphi\cos\qty(\chi+\frac{\pi}4)\vb*{e}_{+1} \\
        + &\sin2\varphi\cos\qty(\chi-\frac{\pi}4)\vb*{e}_{-1} \\
        + &\cos2\varphi\vb*{e}_0,
    \end{split}
    \label{eq:polarization}
\end{align}
\begin{equation*}
    0 \leq \varphi < \pi,\quad -\pi/2 \leq \chi \leq \pi/2,
\end{equation*}
where $\varphi$ represents the azimuth and $\chi$ represents the ellipticity angle. The laser is linearly polarized when $|\chi|=0$ or $\pi/2$, and circularly polarized when $\varphi=\pi/4$ and $\chi=\pm \pi/4$. When this control laser illuminates the atom, the light shift imparted by the laser beam is given by the following expression~\cite{DeutschJessen2010},
\begin{align}
    \label{eq:LS_H}
        &\hat{H}_{LS}  = \frac{3\pi c^2\Gamma}{2\omega_0^3} \times \sum_{q=0,\pm 1}\sum_{F' m_e} \notag \\
        &\qty(\frac{\bra{F m_b}{\vb*{e}^*_L\vdot\hat{\vb*{D}}_{F F'}^{(q)}}\ket{F' m_e} \bra{F' m_e}{\vb*{e}_L\vdot\hat{\vb*{D}}_{F' F}^{\dagger(q)}}\ket{F m_a}}{ \Delta_L - \Delta_{\text{HFS}}(F')})I_L^{(q)} \notag \\
        &\times \ket{F m_b}\bra{F m_a}.
\end{align}
Here, $c$ is the speed of light, $\omega_0$ is the resonant frequency of the $\sSz$-$\tPo$ transition, $\Gamma$ is the natural linewidth of the $\tPo$ state, $\Delta_L$ is the control laser frequency, and $\Delta_{\text{HFS}}(F')$ is the resonant frequency for $F'$. Furthermore, $\hat{\vb*{D}}_{F F'}^{(q)}$ is the orbital annihilation operator, which can be written as
\begin{align}
        \hat{\vb*{D}}^{(q)}_{F F'} = \vb*{e}_q(-1)^{F'+J+1+I}\sqrt{\qty(2F'+1)(2J+1)}\times \notag \\
        \sum_{m_g m_e}\braket{F m_g}{F' m_e; 1 -q}\qty{\mqty{J' & J & 1 \\ F & F' & I}}\ket{F m_g}\bra{F' m_e},
\end{align}
where $\{...\}$ denotes the Wigner $6j$ symbol. The Clebsch-Gordan coefficients are given by
\begin{align}
    &\langle{F m_g | F' m_e; 1 -q \rangle} \notag \\  
    =&(-1)^{F'-1+m_g}\sqrt{2F+1}\mqty(F' & 1 & F \\ m_e & -q & -m_g),
\end{align}
where $\left(...\right)$ is the Wigner $3j$ symbol.

The Hamiltonian Eq.~(\ref{eq:LS_H}) can be engineered using three parameters: the control laser intensity $I_L$, the laser frequency detuned from the transitions' resonance frequency $\Delta_L$, and the laser polarization $\mathbf{e}_L$. In this work, we adjust the laser frequency. Also, since Eq.~(\ref{eq:LS_H}) represents a diagonal matrix, the Hamiltonian can be simply written as Eq.~(\ref{eq:Lightshift}).

%%%%%%%%%%%%%%%%%%%%%%%%%%%%%%%%%%%%%%%%%%%%%%
% Master Equation
%%%%%%%%%%%%%%%%%%%%%%%%%%%%%%%%%%%%%%%%%%%%%%
\subsection{Quantum Master Equation for Spin Dynamics}
\label{method:Master}
The multi-spin dynamics induced by the laser beam used for the qubit manipulation via the single-beam Raman transition technique (Fig.~\ref{fig:dynamics}) are modeled by the following quantum master equation,
\begin{align}
        \pdv{\hat{\rho}(t)}{t} &= -\frac{i}{\hbar}\qty[\hat{H}_{rot},\hat{\rho}(t)] \notag\\ 
        &+ \sum_n\frac12\qty(2\hat{C}_n\hat{\rho}(t)\hat{C}_n^{\dagger} - \hat{\rho}(t)\hat{C}_n^{\dagger}\hat{C}_n - \hat{C}_n^{\dagger}\hat{C}_n\hat{\rho}(t)), 
        \label{eq:lindblad}
\end{align}
\begin{align}
    \hat{H}_{rot} = \hat{d}^{(F)}\qty(\frac{\pi}{2})\hat{H}_{LS}\hat{d}^{(F)\dagger}\qty(\frac{\pi}{2}).
\end{align}
Here, $\hat{C}_n$ represents the collapse operator resulting from photon scattering coming from the control laser light, which is defined as~\cite{DeutschJessen2010}
\begin{align}
    \hat{C}_q = \sqrt{\Gamma}&\sum_{F'}\frac{\Omega/2}{\Delta_L - \Delta_{\text{HFS}}(F')+i\Gamma/2} \notag \\
    &\times \qty(\vb*{e}_q^*\vdot\hat{\vb*{D}}^{(q)}_{FF'})\qty(\vb*{e}_L\vdot\hat{\vb*{D}}_{F'F}^{\dagger(q)}).
    \label{eq:collapse_photon}
\end{align}
For the simulation shown in Fig.~\ref{fig:dynamics}, we employ a fixed control laser detuning of $+11.217\,\text{GHz}$ (for covariant SU(2) rotation) or $-5.005\,\text{GHz}$ (for non-linear rotation) with respect to the $^{1}S_0 $-$ \tPo~(F'=7/2)$ transition frequency. The laser is $\sigma^+$ circularly polarized. To incorporate the imperfection of the initialization, we utilize an experimentally obtained spin population after the initialization to the $\ket{0}_{sc}$ state as $\hat{\rho}(0)$. The laser intensity in the simulation is determined by minimizing the difference between the simulation curve and the experimental data points for the $m_F = -5/2$ state dynamics.

%%%%%%%%%%%%%%%%%%%%%%%%%%%%%%%%%%%%%%%%%%%%%%
% Optimal Condition Proof
%%%%%%%%%%%%%%%%%%%%%%%%%%%%%%%%%%%%%%%%%%%%%%
\subsection{Optimal Differential Lightshift Conditions}
\label{subsec:OptimalProof}
In this section, we derive the necessary differential lightshift (DLS) conditions for realizing the $\hat{R}_x(\pi)$ and $\hat{R}_x^{(cat)}(\pi/2)$ gates, which are given as follows:
\begin{align}
    &\Delta_1:\Delta_2:\Delta_3:\Delta_4:\Delta_5 \notag \\
    &= (2n_1+1):(2n_2+1):(2n_3+1):(2n_4+1):(2n_5+1),
\end{align}
for the $\hat{R}_x(\pi)$ gate, and
\begin{align}
    &\Delta_1:\Delta_2:\Delta_3:\Delta_4:\Delta_5 \notag \\
    &= (4n_1\pm1):(4n_2\mp1):(4n_3\pm1):(4n_4\mp1):(4n_5\pm1),
\end{align}
for the $\hat{R}_x^{(cat)}(\pi/2)$ gate ($n_k \in \mathbb{Z}, ~\text{for}~k=1,2,\dots,2F$).
We then present arguments that these conditions are sufficient to produce the gates.

For notational simplicity, we use the relative lightshifts to that of the $\ket{0}_{sc}$ state throughout this section. That is, we define $\delta'_k = \delta_k - \delta_0~(k=0,1,\dots,2F)$ and subsequently omit the prime notation, treating $\delta'_k$ and $\delta_k$ interchangeably.

\textit{Necessity}-- First, we prove that Eqs.~(\ref{eq:Xgate_condition}) and (\ref{eq:Hgate_condition}) are necessary conditions.
We demand that the unitary time evolution operator $\hat{U}_{rot}^{(F)}(t)$ resulting from the single-beam Raman transition equals to a target unitary operator $\hat{U}$ at some operation time $t$. The operator $\hat{U}$ consists of constants $A$ and $B$ that satisfies the condition of $|A|^2 + |B|^2 = 1$:
\begin{equation}
    \hat{U}(t) = \mqty(B & 0 & 0 & 0 & 0 & A \\
                       0 & B & 0 & 0 & A & 0 \\
                       0 & 0 & B & A & 0 & 0 \\
                       0 & 0 & A & B & 0 & 0 \\
                       0 & A & 0 & 0 & B & 0 \\
                       A & 0 & 0 & 0 & 0 & B).
\end{equation}
The unitary time evolution operator $\hat{U}_{rot}^{(F)}(t)$ can be expressed explicitly expression with Wigner small $d$-matrix~\cite{JJSakurai2020}:
\begin{align}
    d_{m'm}^{(k)}(\beta) &:= \bra{k,m'}\exp\left(-i\frac{\beta}{\hbar}\hat{J}_y\right)\ket{k,m} \notag \\
    &= \sum_{j}(-1)^{j-m+m'} \notag \\
    &\times \frac{\sqrt{(k+m)!(k-m)!(k+m')!(k-m')!}}{(k+m-j)!j!(k-j-m')!(j-m+m')!}\notag \\
    &\times \qty(\cos\frac{\beta}2)^{2k-2j+m-m'}\qty(\sin\frac{\beta}2)^{2j-m+m'},
    \label{eq:wigner_small_d_explicit}
\end{align}
where the summation over $j$ is performed such that all the factorials in the parentheses remain non-negative.
Then, by demanding all matrix elements of $\hat{U}_{rot}^{(F)}(t)$ to coincide with those of $\hat{U}$,
the following equations can be obtained:
\begin{align}
    \label{eq:matrix1}
    &
    \begin{pmatrix}
     -\frac{\sqrt{5}}{32} & -\frac{3\sqrt{5}}{32} & -\frac{2\sqrt{5}}{32} &  \frac{2\sqrt{5}}{32} &  \frac{3\sqrt{5}}{32} &  \frac{\sqrt{5}}{32} \\
      \frac{\sqrt{5}}{32} & -\frac{3\sqrt{5}}{32} &  \frac{2\sqrt{5}}{32} &  \frac{2\sqrt{5}}{32} & -\frac{3\sqrt{5}}{32} &  \frac{\sqrt{5}}{32} \\
     -\frac{\sqrt{10}}{32} &  \frac{\sqrt{10}}{32} &  \frac{2\sqrt{10}}{32} & -\frac{2\sqrt{10}}{32} & -\frac{\sqrt{10}}{32} &  \frac{\sqrt{10}}{32} \\
      \frac{\sqrt{10}}{32} &  \frac{\sqrt{10}}{32} & -\frac{2\sqrt{10}}{32} & -\frac{2\sqrt{10}}{32} &  \frac{\sqrt{10}}{32} &  \frac{\sqrt{10}}{32} \\
     -\frac{5\sqrt{2}}{32} & -\frac{3\sqrt{2}}{32} &  \frac{2\sqrt{2}}{32} & -\frac{2\sqrt{2}}{32} &  \frac{3\sqrt{2}}{32} &  \frac{5\sqrt{2}}{32} \\
      \frac{5\sqrt{2}}{32} & -\frac{3\sqrt{2}}{32} & -\frac{2\sqrt{2}}{32} & -\frac{2\sqrt{2}}{32} & -\frac{3\sqrt{2}}{32} &  \frac{5\sqrt{2}}{32}
    \end{pmatrix}
    \begin{pmatrix}
     1 \\
     e^{-it\delta_1} \\
     e^{-it\delta_2} \\
     e^{-it\delta_3} \\
     e^{-it\delta_4} \\
     e^{-it\delta_5}
    \end{pmatrix}
    \notag \\
    &= 0
    ,
\end{align}
\begin{equation}
    \label{eq:matrix2}
    \begin{pmatrix}
     \tfrac{1}{32} & \tfrac{5}{32} & \tfrac{10}{32} & \tfrac{10}{32} & \tfrac{5}{32} & \tfrac{1}{32} \\
     \tfrac{5}{32} & \tfrac{9}{32} &  \tfrac{2}{32} &  \tfrac{2}{32} & \tfrac{9}{32} & \tfrac{5}{32} \\
     \tfrac{5}{16} & \tfrac{1}{16} & \tfrac{2}{16} & \tfrac{2}{16} & \tfrac{1}{16} & \tfrac{5}{16} \\
     -\tfrac{1}{32} & \tfrac{5}{32} & -\tfrac{10}{32} &  \tfrac{10}{32} & -\tfrac{5}{32} & \tfrac{1}{32} \\
     -\tfrac{5}{32} & \tfrac{9}{32} & -\tfrac{2}{32} &  \tfrac{2}{32} & -\tfrac{9}{32} & \tfrac{5}{32} \\
     -\tfrac{5}{16} & \tfrac{1}{16} & -\tfrac{2}{16} &  \tfrac{2}{16} & -\tfrac{1}{16} & \tfrac{5}{16}
    \end{pmatrix}
    \begin{pmatrix}
    1 \\
    e^{-it\delta_1} \\
    e^{-it\delta_2} \\
    e^{-it\delta_3} \\
    e^{-it\delta_4} \\
    e^{-it\delta_5}
    \end{pmatrix}
    =
    \begin{pmatrix}
    B \\[2pt]
    B \\[2pt]
    B \\[2pt]
    A \\[2pt]
    A \\[2pt]
    A
    \end{pmatrix}
    ,
\end{equation}
Solving for Eq.~(\ref{eq:matrix1}), we obtain
\begin{equation}
     \label{eq:solve_matrix1}
        e^{-it\delta_1}=e^{-it\delta_3}=e^{-it\delta_5},~
        e^{-it\delta_2}=e^{-it\delta_4}=1.
\end{equation}
Substituting Eq.~(\ref{eq:solve_matrix1}) into Eq.~(\ref{eq:matrix2}), we arrive at
\begin{equation}
\label{eq:solution}
\left\{ \,
    \begin{aligned}
        & e^{-it\delta_1} = e^{-it\delta_3} = e^{-it\delta_5} = B+A, \\
        & e^{-it\delta_2} = e^{-it\delta_4} = 1, \\
        & A = B-1.
    \end{aligned}   
\right.
\end{equation}
To perform the $\hat{R}_x(\pi)$ gate, it requires that $B=0$. The DLS $\Delta_{k+1}=\delta_{k+1}-\delta_k~(k=0,1,\dots,2F-1)$ can then be written using integers $n_1, n_2, n_3, n_4, n_5$ as $\Delta_{k+1} = (2n_{k+1} + 1)\pi/t$ (for $k=0,1,\dots,2F-1$).
Therefore, to realize the $\hat{R}_x(\pi)$ gate, the ratio of the five DLSs must satisfy the condition of
\begin{align*}
    &\Delta_1:\Delta_2:\Delta_3:\Delta_4:\Delta_5 \notag \\
    &= (2n_1+1):(2n_2+1):(2n_3+1):(2n_4+1):(2n_5+1).
\end{align*}

For the $\hat{R}_x^{(cat)}(\pi/2)$ gate, it requires that $|A|=|B|$. This leads to the solutions $A=\frac{-1\pm i}{2}, B=\frac{1\pm i}{2}$. From these, the DLSs $\Delta_k$ can be again written using integers as $\Delta_{k+1} = (4n_{k+1} \mp (-1)^{k+1})\pi/t$ (for $k=0,1,\dots,2F-1$). The ratio of the five DLSs must therefore fulfill the condition of
\begin{align*}
    &\Delta_1:\Delta_2:\Delta_3:\Delta_4:\Delta_5 \notag \\
    &= (4n_1\pm1):(4n_2\mp1):(4n_3\pm1):(4n_4\mp1):(4n_5\pm1).
\end{align*}

\textit{Sufficiency}-- Next, we demonstrate that Eqs.~(\ref{eq:Xgate_condition}) and~(\ref{eq:Hgate_condition}) are sufficient conditions. If the five DLSs satisfy Eq.~(\ref{eq:Xgate_condition}), the unitary time evolution $\hat{U}_{rot}^{(F)}(t)$ coincides with $\hat{R}_x(\pi)$ at $t = (2n_1+1)\pi/\Delta_1$. Furthermore, if the five DLSs satisfy Eq.~(\ref{eq:Hgate_condition}), $\hat{U}_{rot}^{(F)}(t)$ coincides with $\hat{R}_x^{(cat)}(\pi/2)$ at $t = (4n_1 \mp 1)\pi/(2\Delta_1)$. 

While these optimal DLS conditions are demonstrated only for $F = 5/2$, it is also applicable to other spin-$F$ systems, such as the $\sSz$ ground state $F = 9/2$ of ${}^{87}\text{Sr}$.

%%%%%%%%%%%%%%%%%%%%%%%%%%%%%%%%%%%%%%%%%%%%%%
% Numerical determination of Optimal Laser Detuning
%%%%%%%%%%%%%%%%%%%%%%%%%%%%%%%%%%%%%%%%%%%%%%
\subsection{Numerical determination of Optimal Laser Detuning}
\label{subsec:detuningsimulation}
Here, we numerically investigate the optimal laser detunings for the above two types of gates under a more realistic setting that includes the effects of photon scattering and deviation from the optimal DLS ratio. 

We calculate the gate infidelities for various laser detunings on the $\sSz$-$\tPo$ transition of $\fermidit$ atoms. 
The gate infidelities are evaluated by simulating the spin state dynamics using the quantum master equation that incorporates the effects of photon scattering (Eq.~(\ref{eq:lindblad})). 
The initial state is set to $\ket{0}_{sc}$ for this simulation. 
The gate fidelity used here is defined as~\cite{steck}:
\begin{equation}
    \mathcal{F}(\hat{\rho}_1, \hat{\rho}_2) := \left| \Tr\sqrt{\sqrt{\hat{\rho}_1}\hat{\rho}_2\sqrt{\hat{\rho}_1}}\right|^2.
    \label{eq:fidelity_definition}
\end{equation}

\begin{figure*}[t]
    \centering
    \includegraphics[width=\textwidth]{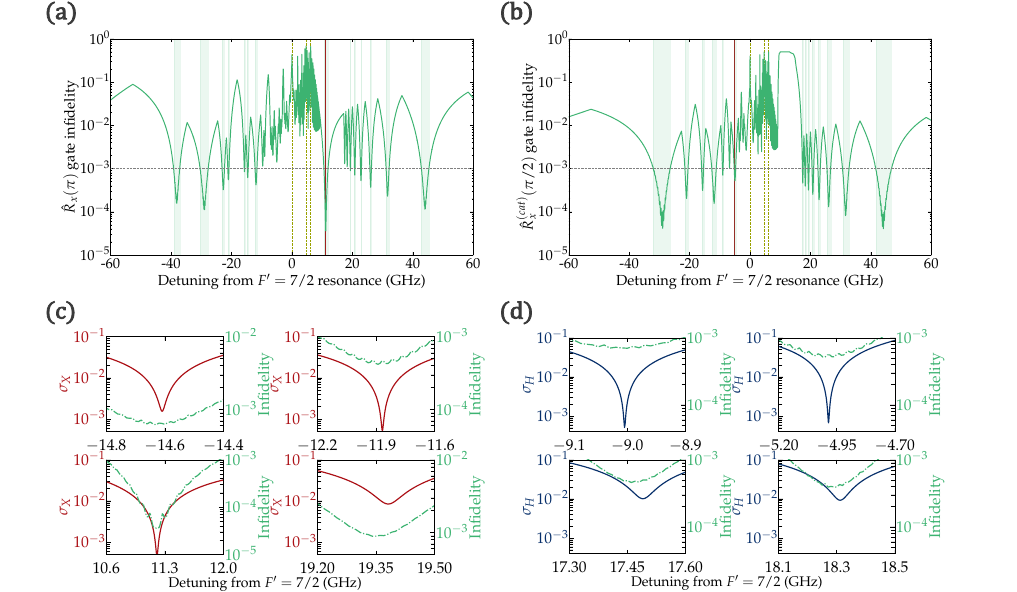}
    \caption{\textbf{Optimal laser detuning for $\hat{R}_x(\pi)$ and $\hat{R}_x^{(cat)}(\pi/2)$ gates on the $\sSz $-$ \tPo$ transition in $\fermidit$.} 
    Dependences of the $\hat{R}_x(\pi)$ gate infidelity (\textbf{a}) and the $\hat{R}_x^{(cat)}(\pi/2)$ gate infidelity (\textbf{b}) on the control laser detuning are plotted as green curves. The red, vertical, solid line for both plots indicates the optimal laser detuning of $+11.217\,\text{GHz}$ and $-5.005\,\text{GHz}$ used in the experiment for the $\hat{R}_x(\pi)$ and $\hat{R}_x^{(cat)}(\pi/2)$ gates, respectively. The infidelity plotted in the figure is extracted from the spin state dynamics calculated using the quantum master equation. The yellow, dashed, vertical lines indicate the resonance frequencies between the $\sSz$ manifold and $\tPo$ manifold for $F'=7/2$, $5/2$, and $3/2$ (from left to right). The frequency ranges where the infidelity is below $10^{-3}$ (indicated by the gray dashed line) are shaded in green.
    \textbf{(c)} Dependences of the $\hat{R}_x(\pi)$ gate infidelity (green dashdot line) and the standard deviation $\sigma_X$ of the DLS ratio (red solid line) on the laser detuning from the $\tPo,F'=7/2$ resonance. 
    \textbf{(b)} Dependences of the $\hat{R}_x^{(cat)}(\pi/2)$ gate infidelity (green dashdot line) and the standard deviation $\sigma_H$ of the DLS ratio (blue solid line) on the laser detuning from the $\tPo,F'=7/2$ resonance. 
    }
    \label{EDfig:optimaldetuning}
\end{figure*}
%----------------1S0 - 3P1, Rx(pi) freq--------------------
\renewcommand{\arraystretch}{1.25}
\begin{table*}[h]
    \centering
    \caption{
    \textbf{Optimal laser detuning for $\hat{R}_x(\pi)$ gate on the $\sSz $-$ \tPo$ transition of $\fermidit$.}\quad Table lists the frequency detuning from the $\sSz$-$\tPo~(F'=7/2)$ resonant frequency that yields the minimum infidelity together with the corresponding infidelity. The frequency detuning region where the gate infidelity is $\le 10^{-3}$, as illustrated in Extended Data Fig.~\ref{EDfig:optimaldetuning}(a), is also listed. The $\hat{R}_x(\pi)$ gate time is specified for a laser power of $100\,\text{mW}$ and a beam waist of $100\,\si{\micro m}$. Furthermore, the ratio of the differential light shifts $\Delta_k$ (for $k=1, \dots, 2F$) induced by the control laser satisfies Eq.~(\ref{eq:Xgate_condition}).
    }
    \begin{tabular}{cccccccc}
        \toprule
         & Index & Detuning (GHz) & Range (GHz) & Infidelity & Gate time (ms) & $\Delta_1:\Delta_2:\Delta_3:\Delta_4:\Delta_5$ &  \\
         \hline
          & 1  & -38.09  & -39.16 $\sim$ -37.14   & $1.57\times10^{-4}$ & 23.43   & 23:27:31:35:39 & \\
          & 2  & -28.91  & -30.08 $\sim$ -27.70   & $1.12\times10^{-4}$ & 6.80    & 11:13:15:17:19 & \\
          & 3  & -22.73  & -23.10 $\sim$ -22.31   & $3.23\times10^{-4}$ & 8.50    & 21:25:29:33:37 & \\
          & 4  & -21.03  & -21.21 $\sim$ -20.83   & $5.49\times10^{-4}$ & 10.97   & 31:37:43:49:55 & \\
          & 5  & -15.696 & -15.770 $\sim$ -15.60  & $8.35\times10^{-4}$ & 6.25    & 29:35:41:47:53 & \\
          & 6  & -14.632 & -14.792 $\sim$ -14.466 & $6.16\times10^{-4}$ & 3.65    & 19:23:27:31:35 & \\
          & 7  & -11.890 & -12.197 $\sim$ -11.572 & $4.19\times10^{-4}$ & 1.24    & 9:11:13:15:17 & \\
          & 8  & 11.217  & 10.595 $\sim$ 11.951   & $3.59\times10^{-5}$ & 0.02697 & 1:1:1:1:1 & \\
          & 9  & 19.354  & 19.318 $\sim$ 19.389   & $8.84\times10^{-4}$ & 3.66    & 27:29:31:33:35 & \\
          & 10 & 20.8    & 20.71 $\sim$ 20.88     & $7.04\times10^{-4}$ & 4.07    & 25:27:29:31:33 & \\
          & 11 & 22.85   & 22.72 $\sim$ 23.00     & $5.26\times10^{-4}$ & 4.74    & 23:25:27:29:31 & \\
          & 12 & 26.05   & 25.81 $\sim$ 26.30     & $3.67\times10^{-4}$ & 5.94    & 21:23:25:27:29 & \\
          & 13 & 31.68   & 31.18 $\sim$ 32.16     & $2.28\times10^{-4}$ & 8.51    & 19:21:23:25:27 & \\
          & 14 & 43.99   & 42.84 $\sim$ 45.34     & $1.15\times10^{-4}$ & 15.92   & 17:19:21:23:25 & \\
         \hline
         \label{EDtab:3P1_173_Rxpi}
    \end{tabular}
\end{table*}

%----------------1S0 - 3P1, Cat freq--------------------
\renewcommand{\arraystretch}{1.25}
\begin{table*}[h]
    \centering
    \caption{
    \textbf{Optimal laser detuning for $\hat{R}_x^{(cat)}(\pi/2)$ gate on the $\sSz $-$ \tPo$ transition in $\fermidit$.}\quad Table lists the frequency detuning from the $\sSz$-$\tPo~(F'=7/2)$ resonant frequency that yields the minimum infidelity together with the corresponding infidelity. The frequency detuning region where the gate infidelity is $\le 10^{-3}$, as illustrated in Extended Data Fig.~\ref{EDfig:optimaldetuning}(b), is also listed. The $\hat{R}_x^{(cat)}(\pi/2)$ gate time is specified for a laser power of $100\,\text{mW}$ and a beam waist of $100\,\si{\micro m}$. Furthermore, the ratio of the differential light shifts $\Delta_k$ (for $k=1, \dots, 2F$) induced by the control laser satisfies Eq.~(\ref{eq:Hgate_condition}).
    }
    \begin{tabular}{cccccccc}
        \toprule
         & Index & Detuning (GHz) & Range (GHz) & Infidelity & Gate time (ms) & $\Delta_1:\Delta_2:\Delta_3:\Delta_4:\Delta_5$ &  \\
         \hline
         & 1  & -28.91  & -32.01  $\sim$ -26.42  & $4.14\times10^{-5}$ & 3.40  & 11:13:15:17:19 & \\
         & 2  & -21.00  & -21.53  $\sim$ -20.51  & $2.03\times10^{-4}$ & 5.47  & 31:37:43:49:55 & \\
         & 3  & -15.696 & -16.018 $\sim$ -15.376 & $3.09\times10^{-4}$ & 3.13  & 29:35:41:47:53 & \\
         & 4  & -11.891 & -12.652 $\sim$ -11.184 & $1.55\times10^{-4}$ & 0.618 & 9:11:13:15:17 & \\
         & 5  & -9.004  & -9.11   $\sim$ -8.903  & $6.92\times10^{-4}$ & 1.11  & 25:31:37:43:49 & \\
         & 6  & -5.005  & -5.227  $\sim$ -4.758  & $5.15\times10^{-4}$ & 0.134 & 7:9:11:13:15 & \\
         & 7  & 17.463  & 17.351  $\sim$ 17.586  & $4.57\times10^{-4}$ & 1.60  & 31:33:35:37:39 & \\
         & 8  & 18.295  & 18.137  $\sim$ 18.440  & $3.87\times10^{-4}$ & 1.70  & 29:31:33:35:37 & \\
         & 9  & 19.354  & 19.156  $\sim$ 19.549  & $3.15\times10^{-4}$ & 1.83  & 27:29:31:33:35 & \\
         & 10 & 20.78   & 20.54.  $\sim$ 21.06   & $2.53\times10^{-4}$ & 2.03  & 25:27:29:31:33 & \\
         & 11 & 22.88   & 22.49   $\sim$ 23.25   & $1.90\times10^{-4}$ & 2.38  & 23:25:27:29:31 & \\
         & 12 & 26.05   & 25.48   $\sim$ 26.66   & $1.30\times10^{-4}$ & 2.97  & 21:23:25:27:29 & \\
         & 13 & 31.61   & 30.63   $\sim$ 32.78   & $8.17\times10^{-5}$ & 4.23  & 19:21:23:25:27 & \\
         & 14 & 44.14   & 41.62   $\sim$ 46.80   & $4.05\times10^{-4}$ & 8.02  & 17:19:21:23:25 & \\
         \hline
         \label{EDtab:3P1_173_Cat}
    \end{tabular}
\end{table*}

As shown in Fig.~\ref{EDfig:optimaldetuning}(a,b), and Tables~\ref{EDtab:3P1_173_Rxpi} and \ref{EDtab:3P1_173_Cat}, we find that there are several detunings where the gate infidelity is significantly suppressed, and confirm that these all satisfy the conditions stated in Eqs.~\ref{eq:Xgate_condition} and \ref{eq:Hgate_condition}. 
These findings strongly support the validity of the derived optimal DLS conditions for realizing high-fidelity single-qubit gates in large-spin systems using single-beam Raman transitions.

%%%%%%%%%%%%%%%%%%%%%%%%%%%%%%%%%%%%%%%%%%%%%%
% Robustness
%%%%%%%%%%%%%%%%%%%%%%%%%%%%%%%%%%%%%%%%%%%%%%
\subsection{Robustness against deviation from optimal differential lightshift ratios}
\label{subsec:robustness}
While the analytical formula demand the DLS ratios to be intergers, realizing a strictly intergral DLS ratios is challenging in practical experiments. To investigate the robustness of the gate fidelity against the deviation from the optimal DLS ratios, we compare the gate infidelities and the deviation around the optimal detuning. Here the deviation is characterized by the standard deviation of the five DLS ratios as follows:
\begin{equation}
    \sigma_{X} = \sqrt{\frac{1}{2F}\sum_{k=2}^{2F}\left(\frac{\Delta_k}{\Delta_1/(2n_1+1)} - (2n_k+1)\right)^2},
\end{equation}
for the $\hat{R}_x(\pi)$ gate, and
\begin{equation}
    \sigma_{H} = \sqrt{\frac{1}{2F}\sum_{k=2}^{2F}\left(\frac{\Delta_k}{\Delta_1/(4n_1\pm1)} - (4n_k\mp(-1)^k)\right)^2},
\end{equation}
for the $\hat{R}_x^{(cat)}(\pi/2)$ gate. 

As shown in Fig.~\ref{EDfig:optimaldetuning}(c,d), we confirm that the standard deviations $\sigma_X$ and $\sigma_H$ are approximately less than $0.01$ around the optimal detunings where the gate infidelities are below $10^{-3}$. While the optimal detunings for minimizing the DLS ratio deviation and the gate infidelities do not perfectly coincide, we can achive high-fidelity gates as long as the laser detuning is within a few hundred MHz of the optimal detunings. This robustness against the DLS ratio deviation relaxes the experimental requirements for implementing high-fidelity single-qubit gates using single-beam Raman transitions in large-spin systems. 

While the DLS ratio deviation characterizes the infidelity of single-qubit gates, photon scattering also significantly contributes to the infidelity. For instance, although the DLS ratio deviations $\sigma_X$ are comparable at detunings of -11.9~GHz and 11.2~GHz, the gate infidelity at -11.9~GHz is approximately one order of magnitude larger than at 11.2~GHz (Fig.~\ref{EDfig:optimaldetuning}(c)). This discrepancy arises from the difference in the number of photon scattering events during the gate.
To quantify the effect of photon scattering, we define the scattering strength as $\Gamma\Omega_L^{2}/(\Delta_L + \Gamma/2)^2$. By multiplying this value by the gate duration, we can estimate the total number of photon scattering events during the operation. For a laser power of 100~mW and a beam waist of 100\,$\si{\micro m}$ with circular polarization, the estimated number of scattering events at -11.9~GHz is $0.052\,\si{kHz} \times 1.24\,\si{ms} = 6.4 \times 10^{-2}$. In contrast, at 11.2~GHz, it is $0.058\,\si{kHz} \times 0.02697\,\si{ms} = 1.6 \times 10^{-3}$. The difference in these values explains the difference in the gate infidelities at these two detunings.

%%%%%%%%%%%%%%%%%%%%%%%%%%%%%%%%%%%%%%%%%%%%%%
% Rank-preserving condition
%%%%%%%%%%%%%%%%%%%%%%%%%%%%%%%%%%%%%%%%%%%%%%
\section{Rank-preserving condition for single-qubit gates}
The implementation of fault-tolerant quantum gates is essential for the realization of large-scale quantum computation. In general, a procedure is defined as {\it fault-tolerant} if a single component failure within the procedure results in at most one error in each encoded block~\cite{NielsenChuang2000}. These components in the procedure include noisy state preparation, noisy gate operations, noisy measurements. While the locality of errors on a specific physical qubit ensures that errors do not propagate to other qubits during single-qubit gates, in spin-cat encoding, certain ideal single-qubit gates may increase the number of hopping errors. This can transform correctable errors into uncorrectable ones, making fault-tolerance difficult to achieve. 
Therefore, it is crucial to design single-qubit gates that do not propagate such hopping errors. In the following, we discuss the {\it rank-preserving condition}, which is a criterion for gates that do not propagate correctable hopping errors into uncorrectable errors. Here, we specifically focus on the rank-preserving nature of single-qubit gates.

We begin by characterizing the correctable errors for the spin-cat code concatenated with a repetition code. For this concatenated code, the set of correctable errors $\mathcal{E}_K$ is given by~\cite{Omanakuttan2024spin}:
\begin{equation}
    \mathcal{E}_K = \textrm{span}\left\{ \hat{S}_q^{(k)}, \hat{A}_q^{(k)} | 0\leq k \leq K, \ -k\leq q \leq k \right\},
    \label{eq:correctable_all}
\end{equation}
where $K = \lfloor \frac{2F-1}{2} \rfloor $ is the maximum number of correctable hopping errors. The operators $\hat{S}_q^{(k)}$ and $\hat{A}_q^{(k)}$ are defined as the following linear combinations of spherical tensor operators:
 \begin{equation}
        \label{eq:OP_operator}
        \begin{aligned}
            & \hat{S}_q^{(k)} = \frac{1}{\sqrt{2}}\qty[\hat{T}_q^{(k)} + (-1)^k\hat{T}_{-q}^{(k)}], \\
            & \hat{A}_q^{(k)} = \frac{1}{\sqrt{2}}\qty[\hat{T}_q^{(k)} - (-1)^k\hat{T}_{-q}^{(k)}], \\
            & \hat{S}_0^{(k)} = \hat{T}_0^{(k)},
        \end{aligned}
\end{equation}
where $0\leq k\leq2F$ and $q=1,2,\dots,k$. 
When applied to the spin-cat states $\ket{\pm}_{F}$, these operators induce hopping and phase-flip errors. Specifically, $\hat{S}_q^{(k)}$ with $q > 0$ induces hopping errors between the spin-cat and kitten states, $\hat{S}_0^{(k)}$ induces phase-flip errors within these states, and $\hat{A}_q^{(k)}$ induces both hopping and phase-flip errors.

The following provides a criterion for gate operations that do not propagate correctable hopping errors into uncorrectable ones.
Given the set of correctable errors $\mathcal{E}_K$ shown in Eq.~\ref{eq:correctable_all}, a gate operation $\hat{U}$ is required to satisfy the following rank-preserving condition:
    \bal
    \hat{U} \hat{E} \hat{U}^\dagger \subset \mathcal{E}_K
    \label{eq:rank preserving def}
    \eal
for every operator $\hat{E}\in\mathcal{E}_K$. 
As shown in Ref.~\cite{Omanakuttan2024spin}, the spin-$F$ SU(2) rotations (i.e. covariant SU(2) rotations) satisfy this rank-preserving condition. The spin-$F$ SU(2) rotational matrix $D_{q'q}^{(k)}(\alpha,\beta,\gamma)$ is given by:
\begin{align}
    &D_{q'q}^{(k)}(\alpha,\beta,\gamma) \notag \\
    &:= \bra{k,q'}\exp\qty(-i\frac{\alpha}{\hbar}\hat{J}_z)\exp\qty(-i\frac{\beta}{\hbar}\hat{J}_y)\exp\qty(-i\frac{\gamma}{\hbar}\hat{J}_z)\ket{k,q} \notag \\
    &= e^{-i(q'\alpha + q\gamma)}\bra{k,q'}\exp\qty(-i\frac{\beta}{\hbar}\hat{J}_y)\ket{k,q}
    \label{eq:wigner_D}
\end{align}
where $\alpha$, $\beta$, and $\gamma$ are the Euler angles, and $\hat{J}_y$ and $\hat{J}_z$ are the spin-$F$ angular momentum operators. This rotational matrix is known as the Wigner D-matrix. When omitting the phase factors, this matrix is same as the wigner small $d$-matrix given as Eq.~(\ref{eq:wigner_small_d_explicit}). 

The proof of the rank-preserving of the spin-$F$ SU(2) rotation is directly derived from the definition of spherical tensor operators~\cite{steck}, 
\begin{equation}
    \hat{D}(\vb*{\zeta}) \hat{T}^{(k)}_q \hat{D}(\vb*{\zeta})^{\dagger} = \sum_{q'=-k}^{k}\hat{T}^{(k)}_{q'}D_{q'q}^{(k)}(\vb*{\zeta}),
    \label{eq:def_sphere}
\end{equation}
where an operator of rank $k$ is transformed into a linear combination of spherical tensor operators of the same rank $k$ under spin-$F$ SU(2) rotations. Consequently, the Pauli gates $\hat{X}$, $\hat{Y}$, and $\hat{Z}$ can be implemented in a rank-preserving manner using spin-$F$ SU(2) rotations:
\begin{equation}
    \hat{X} = \hat{D}(\pi,\pi,0), \quad
    \hat{Y} = \hat{D}(0,\pi,0), \quad
    \hat{Z} = \hat{D}(\pi,0,0).
\end{equation}

While spin-$F$ SU(2) rotations satisfy the rank-preserving condition, it remains unclear whether single-qubit gates that cannot be represented solely by SU(2) rotations can satisfy this criterion. An example of such an operation is the Hadamard gate, as noted in Ref.~\cite{Omanakuttan2024spin}. For completeness, we provide a proof below that fully identifies the set of logical single-qubit gates implementable by a spin-$F$ SU(2) rotation for $F>1/2$. The Hadamard gate is shown to be outside this set.

\begin{pro}
        Suppose $F>1/2$. Then, the set of logical single-qubit gates implemented by rotational operators $\hat{D}(\alpha,\beta,\gamma)$ is generated by $\hat{X}$ and $\hat{R}_z(\theta)$ for an arbitrary $\theta\in[0,2\pi)$. 
\end{pro}

\begin{proof}
We first show that for $\hat{D}(\alpha,\beta,\gamma)$ to be a logical gate, it must hold either $\beta=0$ or $\beta=\pi$. To see this, let $\ket{\psi}$ be the state obtained by applying the rotation operator $\hat{D}(\alpha,\beta,\gamma)$ to an initial state prepared in $\ket{0}=\ket{F,m_F=-F}$, which is given by:
        \bal
            \ket{\psi} &= \sum_{m'_F}e^{-im'_F\alpha + iF\gamma}(-1)^{F+m'_F}\sqrt{\frac{(2F)!}{(F+m'_F)!(F-m'_F)!}}\\
            &\quad \times\qty(\cos\frac{\beta}{2})^{F-m'_F}\qty(\sin\frac{\beta}{2})^{F+m'_F}\ket{F,m'_F}.
        \eal
For this to be a logical state of the spin-cat code, all terms with $m_F'\neq \pm F$ must vanish. 
One can directly see that, if $\beta\neq 0$ or $\beta\neq \pi$, there exists $m_F'\neq \pm F$ (because $F>1/2$) such that the amplitude for such $m_F'$ is non-zero. 
Therefore, it must be the case that $\beta=0$ or $\beta=\pi$ for $\hat{D}(\alpha,\beta,\gamma)$ to be a logical gate. 

It now suffices to analyze the matrix elements $\bra{F,\pm F}\hat{D}(\alpha,\beta,\gamma)\ket{F,\pm F}$ and $\bra{F,\pm F}\hat{D}(\alpha,\beta,\gamma)\ket{F,\mp F}$ for $\beta=0$ and $\beta=\pi$.
Recall that: 
\bal
 \hat{D}(\alpha,\beta,\gamma)\ket{k,m} = \sum_{m'}D_{m'm}^{(k)}(\alpha,\beta,\gamma)\ket{k,m'}
\eal
where $D_{q'q}^{(k)}(\vb*{\zeta})$ for the Euler angles $(\alpha, \beta, \gamma)$ is the Wigner D-matrix defined in Eq.~(\ref{eq:wigner_D}).
When the $\hat{J}_z$ terms of the Wigner D-matrix are omitted, the matrix can also be expressed as Wigner small $d$-matrix $\hat{d}^{(k)}$ defined in Eq.~(\ref{eq:wigner_small_d_explicit}).

With these quantities, the matrix elements of interest can be written as 
\bal
\bra{F,\pm F}\hat{D}(\alpha,\beta,\gamma)\ket{F,\pm F} &= e^{\mp iF(\alpha+\gamma)}d_{\pm F,\pm F}^{(F)}(\beta), \\
\bra{F,\pm F}\hat{D}(\alpha,\beta,\gamma)\ket{F,\mp F} &= e^{\mp iF(\alpha-\gamma)}d_{\pm F, \mp F}^{(F)}(\beta).
\label{eq:rotation operator matrix elements}
\eal
The Wigner small $d$-matrix elements can be explicitly evaluated as 
\bal
 d^{(F)}_{FF}(\beta=0) &= d^{(F)}_{-F,-F}(\beta=0)=1\\
 d^{(F)}_{F,-F}(\beta=0) &= d^{(F)}_{-F,F}(\beta=0)=0\\
 d^{(F)}_{FF}(\beta=\pi) &= d^{(F)}_{-F,-F}(\beta=\pi)=0\\
 d^{(F)}_{F,-F}(\beta=\pi) &= (-1)^{2F} \\
 d^{(F)}_{-F,F}(\beta=\pi)&=1.
\eal
This, together with \eqref{eq:rotation operator matrix elements}, particularly means that when $\beta=0$, the rotation operator becomes a $z$-axis rotational gate $\hat{R}_z(\theta)$ with some angle $\theta$, and an arbitrary $\theta$ can be realized by choosing appropriate $\alpha$ and $\gamma$.
On the other hand, when $\beta=\pi$, the rotation operator becomes a bit flip followed by an additional phase gate, which is written as $\hat{R}_z(\theta)\hat{X}$. In this case too, an arbitrary $\theta$ can be realized by choosing $\alpha$ and $\gamma$ appropriately. 
This means that the set of implementable gates are the ones that can be realized by combining $\hat{R}_z(\theta)$ and $\hat{R}_z(\theta)\hat{X}$, which coincides with the set generated by $\hat{R}_z(\theta)$ and $\hat{X}$.
\end{proof}

Specifically, this indicates that gates capable of creating superpositions, such as the Hadamard gate, cannot be implemented by a covariant SU(2) rotation for $F>1/2$. This conclusion is further supported by the fidelity $\mathcal{F}=\qty|\braket{+}{\psi}|^2$ between the spin-cat state $\ket{+}$ and the state $\ket{\psi}$ obtained by applying the rotation operator $\hat{D}(\alpha,\beta,\gamma)$ to the initial state $\ket{0}=\ket{F,m_F=-F}$:
\begin{equation}
    \mathcal{F} =  \frac12\qty|\qty(\cos\frac{\beta}{2})^{2F} + e^{-2iF\alpha}(-1)^{2F}\qty(\sin\frac{\beta}{2})^{2F}|^2.
\end{equation}
While there exist angles $\alpha$ and $\beta$ in the range $0 \leq \alpha \leq 2\pi$ and $0 \leq \beta \leq 2\pi$ for which $\mathcal{F}=1$ when $F=1/2$, the Hadamard gate cannot be represented by a single rotation operator in systems with dimensions higher than two ($F>1/2$) (Fig.~\ref{fig:hadamard}).

\begin{figure}[t]
    \centering
    \includegraphics[keepaspectratio,width=\linewidth]{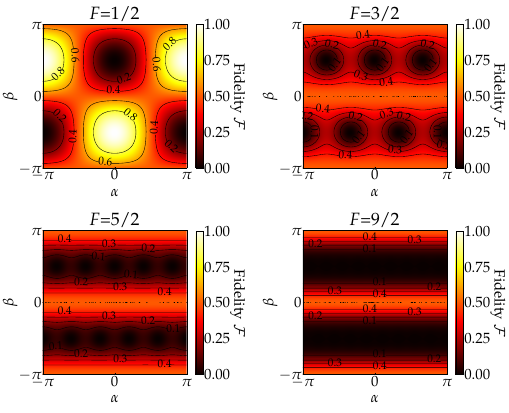}
    \caption[Cat Fidelity solely with SU(2) Operator]{
        \label{fig:hadamard}
        \textbf{Cat Fidelity solely with SU(2) Operator.}\quad Fidelity between the cat state $\ket{+}$ and the state resulting from applying the $(2F+1)$-dimensional rotation operator $\hat{D}(\alpha,\beta,\gamma)$ to the initial state $\ket{0}=\ket{F,m_F=-F}$. The simulation results show that the Hadamard gate can be represented by a single rotation operator only when $F=1/2$.
        }
\end{figure}

Although the Hadamard gate cannot be implemented
solely by an SU(2) rotation, there is still a possibility that it can be implemented in a rank-preserving manner by another technique, e.g., utilizing measurements.
We leave further investigation to future work.

\section{Details of Noise-Bias Dihedral Randomized Benchmarking}
\label{method:DRB}
While CRB is effective for characterizing Clifford gate errors, it cannot be directly applied to non-Clifford operations, such as the $\hat{T}$ gate. Furthermore, to characterize the bias properties of noise, the Hadamard gate must be avoided because it maps $\hat{Z}$ errors to $\hat{X}$ errors, thereby mixing different error types. To address this, dihedral randomized benchmarking (DRB) has been developed to estimate the fidelity of non-Clifford gates~\cite{Dugas2015, cross2016scalable} and to characterize noise-bias structures~\cite{Claes2023, Qing2024FTQC} by utilizing the $\mathbb{D}_8$ dihedral group. The single-qubit $\mathbb{D}_8$ dihedral group is generated by the $\hat{T}$ and $\hat{X}$ gates. The DRB protocol follows a procedure similar to CRB, with the key difference being the use of random sequences from the $\mathbb{D}_8$ group instead of the Clifford group. For noise-bias characterization, the DRB protocol involves two separate experiments: one with preparation and measurement in the $z$-basis and the other in the $x$-basis. The steps for the single-qubit noise-bias DRB protocol are as follows~\cite{Claes2023}:
\begin{enumerate}
    \item Prepare the qubit in an eigenstate of the $\hat{Z}$ (or $\hat{X}$) operator, typically $\ket{0}$ (or $\ket{+}$).
    \item Apply a unitary operation $\hat{P}$ randomly sampled from the single-qubit Pauli group. 
    \item Apply a sequence of $m$ gates randomly sampled from the $\mathbb{D}_8$ dihedral group. 
    \item Apply an inverse unitary that returns the qubit to the initial state.
    \item Measure the final state in the $z$- (or $x$-) basis and record 1 if the measurement outcome matches the initial state, and 0 otherwise.
    \item Repeat the steps above for various sequence lengths $m+2$.
    \item Average the outcomes to determine the survival probability for each length $m+2$.
\end{enumerate}
The two DRB circuits for the $z$-basis and the $x$-basis measurements are illustrated in Fig.~\ref{fig:bias}(a). Note that the original noise-bias DRB protocol~\cite{Claes2023} includes only the inverse gate for the length-$m$ random sequence presented in Step 3, while post-processing the measurement results by a factor of ±1 depending on the initially applied Pauli gate $\hat{P}$. To simplify the experimental implementation, we include the inverse of the entire sequence, up to $\hat{P}$. 

The average return probabilities $P_1(m)$ for $z$-basis DRB and $P_2(m)$ for $x$-basis DRB follow exponential decay models:
\begin{equation}
    P_1(m) = A_1 \lambda_1^{m+2} + B_1, \quad P_2(m) = A_2 \lambda_2^{m+4} + B_2,
\end{equation}
where coefficients $A_j$ and $B_j$ absorb state preparation and measurement errors. In this work, we fix $B_1$ and $B_2$ to the population values measured in Fig.~\ref{EDfig:SSR}(e) to reduce the fitting uncertainty.
The decay parameters $\lambda_j$ are related to the average dephasing error probability $p_{D}$ and the average non-dephasing error probability $p_{ND}$ by:
\begin{align}
    p_{D} &= \frac{2^N - 1}{4^N}\left\{1 + (2^N-1)\lambda_1 - 2^N\lambda_2 \right\}, \\
    p_{ND} &= \frac{2^N - 1}{2^N}\left(1 - \lambda_1\right),
\end{align}
where $2^N$ is the Hilbert space dimension ($N=1$ for a single qubit). The bias parameter $\eta$ is defined as the ratio: 
\begin{equation}
    \eta = \frac{p_{D}}{p_{ND}}.
\end{equation}

To implement the DRB protocol, we identify the complete set of gates comprising the $\mathbb{D}_8$ dihedral group. This group consists of 22 distinct elements generated by $\hat{T}$ and $\hat{X}$ gates. A comprehensive list of these 22 gates, including their decompositions into $\hat{R}_z(\pi/4)$ and $\hat{R}_x(\pi)$ rotations, is provided in Table~\ref{tab:dihedral_gate_list}. In our DRB implementation, the average number of pulses per dihedral gate is 1.0(3).

\begin{table*}[h]
    \centering
    \caption{
        \textbf{Single-qubit $\mathbb{D}_8$ dihedral gates and their decompositions.} 
        Single-qubit $\mathbb{D}_8$ dihedral gates are presented alongside their corresponding matrix representations and decompositions into $\hat{R}_z(\pi/4)$ and $\hat{R}_x(\pi)$ rotations. Global phase factors are omitted. For simplicity, $\hat{R}_z(\theta)$ and $\hat{R}_x(\theta)$ are denoted as $Z(\theta)$ and $X(\theta)$, respectively.
        }
    \setlength{\tabcolsep}{7pt}
    \begin{tabular}{cccc|cccc}
        \toprule
        Index & Label & Gate matrix & Pulse decomposition & Index & Label & Gate matrix & Pulse decomposition \\
        \midrule
        1  & $\hat{I}$ & $\mqty(1 & 0 \\ 0 & 1)$ & - & 11 & $\hat{X}\hat{T}$ & $\mqty(0 & e^{i\pi/4} \\ 1 & 0)$ & $X(\pi)Z(\frac{\pi}{4})$ \\
        2  & $\hat{X}$ & $\mqty(0 & 1 \\ 1 & 0)$ & $X(\pi)$ & 12 & $\hat{X}\hat{S}$ & $\mqty(0 & i \\ 1 & 0)$ & $X(\pi)Z(\frac{\pi}{2})$ \\
        3  & $\hat{Y}$ & $\mqty(0 & -i \\ i & 0)$ & $X(\pi)Z(\pi)$ & 13 & $\hat{X}\hat{T}\hat{S}$ & $\mqty(0 & e^{i3\pi/4} \\ 1 & 0)$ & $X(\pi)Z(\frac{3\pi}{4})$ \\
        4  & $\hat{Z}$ & $\mqty(1 & 0 \\ 0 & -1)$ & $Z(\pi)$ & 14 & $\hat{X}\hat{T}^\dagger$ & $\mqty(0 & e^{-i\pi/4} \\ 1 & 0)$ & $X(\pi)Z(\frac{7\pi}{4})$ \\
        5  & $\hat{T}$ & $\mqty(1 & 0 \\ 0 & e^{i\pi/4})$ & $Z(\frac{\pi}{4})$ & 15 & $\hat{X}\hat{S}^\dagger$ & $\mqty(0 & -i \\ 1 & 0)$ & $X(\pi)Z(\frac{3\pi}{2})$ \\
        6  & $\hat{S}$ & $\mqty(1 & 0 \\ 0 & i)$ & $Z(\frac{\pi}{2})$ & 16 & $\hat{X}\hat{T}^\dagger\hat{S}^\dagger$ & $\mqty(0 & e^{-i3\pi/4} \\ 1 & 0)$ & $X(\pi)Z(\frac{5\pi}{4})$ \\
        7  & $\hat{T}\hat{S}$ & $\mqty(1 & 0 \\ 0 & e^{i3\pi/4})$ & $Z(\frac{3\pi}{4})$ & 17 & $\hat{T}\hat{X}$ & $\mqty(0 & e^{-i\pi/4} \\ 1 & 0)$ & $Z(\frac{\pi}{4})X(\pi)$ \\
        8  & $\hat{T}^\dagger$ & $\mqty(1 & 0 \\ 0 & e^{-i\pi/4})$ & $Z(\frac{7\pi}{4})$ & 18 & $\hat{S}\hat{X}$ & $\mqty(0 & -i \\ 1 & 0)$ & $Z(\frac{\pi}{2})X(\pi)$ \\
        9  & $\hat{S}^\dagger$ & $\mqty(1 & 0 \\ 0 & -i)$ & $Z(\frac{3\pi}{2})$ & 19 & $\hat{T}\hat{S}\hat{X}$ & $\mqty(0 & e^{-i3\pi/4} \\ 1 & 0)$ & $Z(\frac{3\pi}{4})X(\pi)$ \\
        10 & $\hat{T}^\dagger\hat{S}^\dagger$ & $\mqty(1 & 0 \\ 0 & e^{-i3\pi/4})$ & $Z(\frac{5\pi}{4})$ & 20 & $\hat{T}^\dagger\hat{X}$ & $\mqty(0 & e^{i\pi/4} \\ 1 & 0)$ & $Z(\frac{7\pi}{4})X(\pi)$ \\
         & & & & 21 & $\hat{S}^\dagger\hat{X}$ & $\mqty(0 & i \\ 1 & 0)$ & $Z(\frac{3\pi}{2})X(\pi)$ \\
         & & & & 22 & $\hat{T}^\dagger\hat{S}^\dagger\hat{X}$ & $\mqty(0 & e^{3i\pi/4} \\ 1 & 0)$ & $Z(\frac{5\pi}{4})X(\pi)$ \\
        \bottomrule
    \end{tabular}
    \label{tab:dihedral_gate_list}
\end{table*}

%%%%%%%%%%%%%%%%%%%%%%%%%%%%%%%%%%%%%%%%%%%%%%
% Construction of Error Budget
%%%%%%%%%%%%%%%%%%%%%%%%%%%%%%%%%%%%%%%%%%%%%%
\section{Construction of Error Budgets}
\label{sec:error_budget_detail}
The effects of the error sources on the single-qubit gate errors are extracted from simulated randomized benchmarking (RB). In this RB, the dynamics of each gate pulse are sampled using the quantum master equation.
The single-qubit gate error sources considered in our analysis are listed below with the corresponding Clifford gate errors, the non-dephasing probability $p_{ND}$, and the dephasing probability $p_{D}$ listed in Table~\ref{tab:error_budget}.

\begin{itemize}
    \item \textit{Laser Intensity Fluctuation.} The laser intensity fluctuation is monitored by measuring the fluctuation of the pulse area. The measured standard deviation $\sigma$ of the pulse area is 0.2\% of the mean pulse area for QB1 (used for $\text{SU}(2)$ rotation) and 0.56\% for QB2 (used for cat generation pulse) over the course of the experiment of about 40 minutes.

    \item \textit{Laser Polarization Fluctuation.} The laser polarization fluctuation is monitored using a polarimeter (PAX1000, Thorlabs) by measuring the standard deviation of the azimuth $\varphi$ and ellipticity $\chi$ (Eq.~(\ref{eq:polarization})). The measured standard deviations for $\varphi$ and $\chi$ for QB1 (QB2) are $3.1\,^{\circ}$ ($0.8\,^{\circ}$) and $0.13\,^{\circ}$ ($0.02\,^{\circ}$), respectively.

    \item \textit{Laser Frequency Fluctuation.} The control laser frequency is stabilized by a wavemeter (WS8-10, HighFinesse), and the laser linewidth (standard deviation, $\sigma$) is less than $1\,\text{MHz}$ for both beams QB1 and QB2. In the simulation, the laser frequency is sampled from a Gaussian distribution with $\sigma=1\,\text{MHz}$.

    \item \textit{Orthogonality Imperfection.} The deviation of the crossing angle between QB1 and QB2 from $90\,^{\circ}$, derived from the optical path design, is $2\,^{\circ}$. For the CRB experiment, we align the magnetic field parallel to the QB1 propagation axis, and then only QB2 is deviated by $2\,^{\circ}$. For the DRB experiment, the magnetic field is applied orthogonal to both QB1 and QB2, and the simulation sets the misalignment such that both QB1 and QB2 are $2\,^{\circ}$ away from the magnetic field direction.

    \item \textit{Finite Zeeman Splitting.} The single-beam Raman transition assumes that the magnitude of the light shift induced by the control laser is sufficiently larger than the Zeeman splitting. A finite Zeeman splitting causes deviations from the ideal rotation achieved by the single-beam Raman transition. We simulate this effect on the gate fidelity via the master equation. In our experiment, a magnetic field of $1.01\,\text{mT}$ is applied for the CRB measurement, and $1.35\,\text{mT}$ for the DRB measurement. In the simulation, the master equation calculation is modified by using the Hamiltonian,
    \begin{align}
        \hat{H}_{rot} = &\hat{d}^{(F)}\qty(\frac{\pi}{2})\hat{H}_{LS}\hat{d}^{(F)\dagger}\qty(\frac{\pi}{2}) \notag \\
        &+ \mu_Bg_FB\sum_{k=0}^{2F}(-F+k)\ket{-F+k}\bra{-F+k},
    \end{align}
    where $B$ is the magnitude of the magnetic field, $\mu_B$ is the Bohr magneton and the Land\'e g-facor is calculated as $g_F=-\mu_I/(\mu_B|I|)$~\cite{kroeze2025171}. The nuclear magnetic moment  $\mu_I$ of the $\sSz$ state of $\fermidit$ is $-0.6776\,\mu_N$~\cite{Porsev2004}, where $\mu_N$ is the nuclear magneton.

    \item \textit{Dephasing.} A finite coherence time sets a upper bound on the accuracy with which any quantum gate can be executed. To simulate the upper limit of from this, we use the coherence times $T_{2,k}^*=251(21)~\text{ms}/(F-k)~(k=0,1,2)$ obtained from the measurements shown in Fig.~\ref{fig:lifetime}. The master equation is calculated by setting the collapse operator $\hat{C}$ as:
    \begin{align}
        \hat{C} = \sum_{k=0}^{\frac{2F-1}{2}}\sqrt{\frac{1}{T_{2,k}^*}}\qty(\ket{-F+k}\bra{-F+k} - \ket{F-k}\bra{F-k}).
    \end{align}

    \item \textit{Photon Scattering.} To account for the effect of photon scattering caused by the control laser light, the master equation is calculated by setting the collapse operator as Eq.~(\ref{eq:collapse_photon}).
\end{itemize}

\renewcommand{\arraystretch}{1.25}
\begin{table*}
    \centering
    \caption{List of error sources on the single-qubit gate with the corresponding Clifford gate error the source produces, the non-dephasing probability $p_{ND}$, and the dephasing probability $p_{D}$.}
    \begin{tabular}{l|c|c|c}
        \toprule
         Error source & Clifford gate error & $p_{ND}$ & $p_{D}$ \\ \hline
         Laser intensity fluctuation    & $7.6(3)\times10^{-3}$ & $5.5(1)\times 10^{-8}$ & $8.1(4)\times10^{-6}$ \\
         Laser polarization fluctuation & $5.6(3)\times10^{-2}$ & $1.8(1)\times 10^{-3}$ & $4.8(3)\times 10^{-3}$ \\
         Laser frequency fluctuation    & $7.1(3)\times10^{-5}$ & $5.1(2)\times 10^{-8}$ & $2.1(3)\times 10^{-6}$ \\
         Orthogonality imperfection     & $3.4(3)\times10^{-3}$ & $6.8(7)\times 10^{-4}$ & $2.0(2)\times 10^{-3}$ \\
         Finite Zeeman splitting        & $5.0(4)\times10^{-3}$ & $3.1(3)\times 10^{-4}$ & $1.2(1)\times 10^{-3}$ \\
         Dephasing                      & $5.8(1)\times10^{-4}$ & $3.1(0)\times 10^{-7}$ & $7.5(1)\times 10^{-4}$ \\
         Photon scattering              & $7.8(1)\times10^{-4}$ & $7.7(1)\times 10^{-8}$ & $3.4(0)\times 10^{-5}$ \\
        \hline
    \end{tabular}
    \label{tab:error_budget}
\end{table*}

Fig.~\ref{fig:budget} in the main text shows the comparison between the contribution of each error source and the experimental results. We anticipate that a gate fidelity exceeding $0.999$ can be achieved through technical upgrades, including stabilization of the laser intensity and polarization, improvement in the orthogonality alignment, and the reduction of the finite Zeeman splitting contribution by realizing faster gate operations.

%%%%%%%%%%% References %%%%%%%%%%
\bibliography{Refs}

%%%%%%%%%%% Extended Data %%%%%%%%%%
\setcounter{figure}{0}

\end{document}